\documentclass[runningheads]{llncs}
\usepackage{amssymb}
\usepackage{amsmath}
\usepackage{comment}
\usepackage[disable]{todonotes}
\usepackage{inputenc}
\usepackage{shuffle}
\usepackage{multirow}
\usepackage{listings}
\usepackage{mathabx}
\usepackage{hyperref}

\usepackage{shuffle}

\usepackage{tikz}
\usetikzlibrary{positioning,shadows,arrows}
\usetikzlibrary{arrows,automata,positioning}
\usetikzlibrary{shapes,shapes.geometric,arrows,fit,calc,positioning,automata,}

\usepackage{diagbox}
\usepackage{slashbox}
\usepackage{tikz}
\usetikzlibrary{matrix}

\newcommand{\stc}{\operatorname{sc}}

\usepackage{apxproof}
\theoremstyle{plain}
\newtheoremrep{theorem}{Theorem}%[section]
\newtheoremrep{proposition}[theorem]{Proposition}
\newtheoremrep{lemma}[theorem]{Lemma}
\newtheoremrep{claim}[theorem]{Claim}
\newtheoremrep{conjecture}[theorem]{Conjecture}
\newtheoremrep{corollary}[theorem]{Corollary}
\newtheoremrep{definition}[theorem]{Definition}

\newenvironment{claiminproof}[1]{\medskip\par\noindent\underline{Claim:}\space#1}{}
\newenvironment{claimproof}[1]{\begin{quote}\par\noindent\emph{Proof of the Claim:}\space#1}{[\emph{End, Proof of the Claim}]\end{quote}}%\newline}
\DeclareMathOperator{\lcm}{lcm}

\DeclareFontFamily{U}{bigshuffle}{}
\DeclareFontShape{U}{bigshuffle}{m}{n}{
  <5-8> s*[1.7] shuffle7
  <8->  s*[1.7] shuffle10
}{}
\DeclareSymbolFont{BigShuffle}{U}{bigshuffle}{m}{n}
\DeclareMathSymbol\bigshuffle{\mathop}{BigShuffle}{"001}
\DeclareMathSymbol\bigcshuffle{\mathop}{BigShuffle}{"002}

\begin{document}
\title{Commutative Regular Languages with Product-Form Minimal Automata}

\author{Stefan Hoffmann\orcidID{0000-0002-7866-075X}}
\authorrunning{S. Hoffmann}
\institute{Informatikwissenschaften, FB IV, 
  Universit\"at Trier,  Universitätsring 15, 54296~Trier, Germany, 
  \email{hoffmanns@informatik.uni-trier.de}}
\maketitle              % typeset the header of the contribution
\begin{abstract}
 We introduce a subclass of the commutative regular languages
 that is characterized by the property that
 the state set of the minimal deterministic automaton can be written as a certain Cartesian product.
 This class behaves much better with respect to the state complexity
 of the shuffle, for which we find the bound~$2nm$ if the input languages have state complexities $n$ and $m$,
 and the upward and downward closure and interior operations,
 for which we find the bound~$n$.
 In general, only the bounds $(2nm)^{|\Sigma|}$
 and $n^{|\Sigma|}$ are known for these operations in the commutative case.
 We prove different characterizations of this class
 and present results to construct languages from this class.
 Lastly, in a slightly more general setting of partial commutativity, we introduce other, related, 
 language classes and investigate the relations between them.

\keywords{finite automaton \and state complexity \and shuffle \and upward closure \and downward closure \and commutative language \and product-form minimal automaton \and partial commutation} 
\end{abstract}

\section{Introduction}
\label{sec:introduction}

The state complexity, as used here, of a regular language $L$
is the minimal number of states needed in a complete deterministic automaton
recognizing~$L$. The state complexity of an operation
on regular languages is the greatest state complexity
of the result of this operation
as a function of the (maximal) state complexities of its arguments.

Investigating the state complexity of the result of a regularity-preserving operation on regular languages,
see~\cite{GaoMRY17} for a survey, was first initiated by Maslov in~\cite{Mas70} and systematically started by Yu, Zhuang \& Salomaa in~\cite{YuZhuangSalomaa1994}.

A language is called commutative, if
for each word in the language, every permutation of this word is also in the language.
The class of commutative automata, which recognize commutative regular languages, was introduced in~\cite{BrzozowskiS73}.
% Here,
% we consider the state complexity of the projections, the shuffle operation, of union and intersection,
% and of the upward and downward closure and interiors 
% on the class of commutative regular languages.
% The reader might notice that several results, notions and methods of proof
% are generalizations from unary languages, and, in fact, viewing commutative
% languages as generalizing unary languages is a guiding theme.

The shuffle and iterated shuffle have been introduced and studied to understand %, or specify,
the semantics of parallel programs. This was undertaken, as it appears
to be, independently by Campbell and Habermann~\cite{CamHab74}, by Mazurkiewicz~\cite{DBLP:conf/mfcs/Marzurkiewicz75}
and by Shaw~\cite{Shaw78zbMATH03592960}. They introduced \emph{flow expressions}, 
which allow for sequential operators (catenation and iterated catenation) as well
as for parallel operators (shuffle and iterated shuffle)
to specify sequential and parallel execution traces.%, also see~\cite{DBLP:journals/cl/Riddle79}.

The shuffle operation as a binary operation, but not the iterated shuffle,
is regularity-preserving on all regular languages. The state complexity of the shuffle
operation in the general cases was investigated in~\cite{BrzozowskiJLRS16}
for complete deterministic automata and in~\cite{DBLP:journals/jalc/CampeanuSY02}
for incomplete deterministic automata. The bound $2^{nm-1} + 2^{(m-1)(n-1)}(2^{m-1}-1)(2^{n-1}-1)$ was obtained
in the former case, which is not known to be tight, and the tight bound $2^{nm}-1$
in the latter case. 

% higman, haines, simon piecewise testable, holzer zitieren?
A word is a (scattered) subsequence of another word, if it can be obtained from the latter word by deleting
letters. This gives a partial order, and the upward and downward closure and interior operations
refer to this partial order. The upward closures are also known as shuffle ideals.
The state complexity of these operations was investigated in~\cite{DBLP:journals/tcs/GruberHK07,DBLP:journals/fuin/GruberHK09,DBLP:journals/ita/Heam02,KarandikarNS16,DBLP:journals/fuin/Okhotin10}

The state complexity of the projection operation was investigated in~\cite{Hoffmann2021DLT,DBLP:journals/tcs/JiraskovaM12,Wong98}.
In~\cite{Wong98}, the tight upper bound $3 \cdot 2^{n-2} - 1$
was shown, and in~\cite{DBLP:journals/tcs/JiraskovaM12} the refined, and tight, bound $2^{n-1} + 2^{n-m} - 1$
was shown, where $m$ is related to the number of unobservable transitions for the projection operator.
Both results were established for incomplete deterministic automata.

In~\cite{Hoffmann2021NISextended,DBLP:conf/cai/Hoffmann19,DBLP:conf/dcfs/Hoffmann21,Hoffmann2021DLT} the state complexity of these operations
was investigated for commutative regular languages.
The results are summarized in Table~\ref{tab:sc_known_results}.

\begin{table}[ht]
    \centering % sharpness für projektionen
    \begin{tabular}{|c|c|c|c|} %bei sharpness mit statement. sonst open problem, general bound  %mit zitiren.
    % statement aussage + sharpness
    \hline
    Operation                                    &  Upper Bound        &  Lower Bound              & References \\ \hline %& General Bound \\ \hline
    $\pi_{\Gamma}(U)$, $\Gamma \subseteq \Sigma$ & $n$                 &  $n$  &  \cite{Hoffmann2021NISextended,Hoffmann2021DLT} \\
    $U \shuffle V$ & $\min\{(2nm)^{|\Sigma|},f(n,m)\}$ & $\Omega\left( nm \right) $ &  \cite{BrzozowskiJLRS16,Hoffmann2021NISextended,DBLP:conf/cai/Hoffmann19} \\
    $\uparrow\! U$                               & $n^{|\Sigma|}$ & $\Omega\left( \left( \frac{n}{|\Sigma|} \right)^{|\Sigma|} \right)$ & \cite{DBLP:journals/ita/Heam02,Hoffmann2021NISextended,DBLP:conf/dcfs/Hoffmann21} \\
    $\downarrow\! U$                             & $n^{|\Sigma|}$ & $n$ &  \cite{Hoffmann2021NISextended,DBLP:conf/dcfs/Hoffmann21} \\
    $\uptodownarrow\! U$                         & $n^{|\Sigma|}$ & $\Omega\left( \left( \frac{n}{|\Sigma|} \right)^{|\Sigma|} \right)$ &  \cite{Hoffmann2021NISextended,DBLP:conf/dcfs/Hoffmann21} \\
    $\downtouparrow\! U$                         & $n^{|\Sigma|}$ & $n$ & \cite{Hoffmann2021NISextended,DBLP:conf/dcfs/Hoffmann21} \\
    %$U \cap V$                                   & $nm$               & sharp, for any $\Sigma$ & \cite{Hoffmann20,Hoffmann2021} \\
    $U \cup V$,$U \cap V$                                  & $nm$               & tight for each $\Sigma$ & \cite{Hoffmann2021NISextended,DBLP:conf/cai/Hoffmann19} \\ \hline
    \end{tabular}
    \caption{Overview of results for commutative regular languages.
     The state complexities of the input languages are $n$ and $m$.
     Also, $f(n,m) =  2^{nm-1} + 2^{(m-1)(n-1)}(2^{m-1}-1)(2^{n-1}-1)$
     is the general bound for shuffle from~\cite{BrzozowskiJLRS16} in case of complete automata.} 
    \label{tab:sc_known_results}
\end{table}

\begin{table}[ht]
    \centering % sharpness für projektionen
    \begin{tabular}{|c|c|c|c|} %bei sharpness mit statement. sonst open problem, general bound  %mit zitiren.
    % statement aussage + sharpness
    \hline
    Operation                                    &  Upper Bound        &  Lower Bound & Reference             \\ \hline 
    $\pi_{\Gamma}(U)$, $\Gamma \subseteq \Sigma$ & $n$                 &  $n$   & Thm.~\ref{thm:sc_results} \\
    $U \shuffle V$ &  $2nm$ & $\Omega\left( nm \right) $ & Thm.~\ref{thm:sc_results} \\
    $\uparrow\! U$,$\downarrow\! U$,$\uptodownarrow\! U$,$\downtouparrow\! U$                                & $n$ & $n$ & Thm.~\ref{thm:sc_results}  \\
    %$\downarrow\! U$                             & $n$ & $n$ \\
    %$\uptodownarrow\! U$                         & $n$ & $n$ \\
    %$\downtouparrow\! U$                         & $n$ & $n$  \\
    $U \cap V$, $U \cup V$                                    & $nm$               & tight for each $\Sigma$ & Thm.~\ref{thm:sc_results} \\ \hline
    %$U \cup V$                                   & $nm$               & sharp, for any $\Sigma$ \\ \hline
    \end{tabular}
    \caption{State Complexity results on the subclass
    of commutative languages with product-form minimal automaton
    for input languages with state complexities $n$ and~$m$.}
    \label{tab:sc_product-from}
\end{table}

 In~\cite{GomezA08} the minimal commutative automaton was introduced, which can be associated
 with every commutative regular language.\todo{die unteren schranken noch zeigen.}
 This automaton played a crucial role in~\cite{Hoffmann2021NISextended,DBLP:conf/cai/Hoffmann19}
 to derive the bounds mentioned in Table~\ref{tab:sc_known_results}.
 Here, we will investigate the subclass of those language
 for which the minimal commutative automaton is in fact the smallest automaton
 recognizing a given commutative language.\todo{und beispiel unterschied kann beliebig groß werden.}
 For this language class, we will derive the following state complexity bounds
 summarized in Table~\ref{tab:sc_product-from}.
Additionally, we will prove other characterizations
and properties of the subclass considered and relate it with other subclasses, in a more general setting,
in the final chapter.%  of special form of partial commutativity.

\section{Preliminaries}

In this section and Section~\ref{sec:product-form},
we assume that $k \ge 0$ denotes our alphabet size and 
$\Sigma = \{a_1, \ldots, a_k\}$ is our \emph{alphabet}.
We will also write $a,b,c$ for $a_1,a_2,a_3$ in case of $|\Sigma| \le 3$.
The set $\Sigma^{\ast}$ denotes
the set of all finite sequences over $\Sigma$, i.e., of all \emph{words}. The finite sequence of length zero,
or the \emph{empty word}, is denoted by $\varepsilon$. For a given word we denote by $|w|$
its \emph{length}, and for $a \in \Sigma$ by $|w|_a$ the \emph{number of occurrences of the symbol $a$}
in $w$. 
%Also for $\Gamma \subseteq \Sigma$, $\Gamma^*$
%is the set of \emph{finite sequences over $\Gamma$}.
For $a \in \Sigma$, we set $a^* = \{a\}^*$.
A \emph{language} is a subset of $\Sigma^*$.
For $u \in \Sigma^*$, the \emph{left quotient} is $u^{-1}L = \{ v \in \Sigma^* \mid uv \in L\}$
and the \emph{right quotient} is $Lu^{-1} = \{ v \in \Sigma^* \mid vu \in L \}$.

%We assume the reader
%to have some basic knowledge in formal language
%theory, as contained, e.g., in~\cite{HopUll79}. For instance, we make use of regular expressions to describe
%languages

The \emph{shuffle operation}, denoted by $\shuffle$, is defined by
 \begin{multline*}
    u \shuffle v  = \{ w \in \Sigma^*  \mid  w = x_1 y_1 x_2 y_2 \cdots x_n y_n 
    \emph{ for some words } \\ x_1, \ldots, x_n, y_1, \ldots, y_n \in \Sigma^*
    \emph{ such that } u = x_1 x_2 \cdots x_n \emph{ and } v = y_1 y_2 \cdots y_n \},
 \end{multline*}
 for $u,v \in \Sigma^{\ast}$ and 
  $L_1 \shuffle L_2  := \bigcup_{x \in L_1, y \in L_2} (x \shuffle y)$ for $L_1, L_2 \subseteq \Sigma^{\ast}$.
 %The shuffle operation is commutative, associative and distributive with respect
 %to union. We will use these properties without further mention. 
%  In writing formulas
%  without brackets, we suppose that the shuffle operation binds stronger than the set operations,
%  and the concatenation operator has the strongest binding. 
 If $L_1, \ldots, L_n \subseteq \Sigma^*$, we set $\bigshuffle_{i=1}^n L_i = L_1 \shuffle \ldots \shuffle L_n$.

Let $\Gamma \subseteq \Sigma$.
The \emph{projection homomorphism} $\pi_{\Gamma} : \Sigma^* \to \Gamma^*$ 
is given by $\pi_{\Gamma}(x) = x$ for $x \in \Gamma$
and $\pi_{\Gamma}(x) = \varepsilon$ for $x \notin \Gamma$ and extended
to $\Sigma^*$ by $\pi_{\Gamma}(\varepsilon) = \varepsilon$
and $\pi_{\Gamma}(wx) = \pi_{\Gamma}(w)\pi_{\Gamma}(x)$ for $w \in \Sigma^*$ and $x \in \Sigma$.
As a shorthand, we set, with respect to a given naming $\Sigma = \{a_1, \ldots, a_k\}$,
$\pi_j = \pi_{\{a_j\}}$. Then $\pi_j(w) = a_j^{|w|_{a_j}}$.
%A \emph{non-trivial projection}
%is a projection with $\Gamma$ non-empty
%and not equal $\Sigma$.

A language $L \subseteq \Sigma^*$ is \emph{commutative},
if, for $u,v \in \Sigma^*$ such that $|v|_x = |u|_x$ for every $x \in \Sigma$,
we have $u \in L$ if and only if $v \in L$, i.e., $L$
is closed under permutation of letters in words from $L$.

A quintuple $\mathcal A = (\Sigma, Q, \delta, q_0, F)$ is a \emph{finite deterministic and complete automaton} (DFA),  
where $\Sigma$ is the \emph{input alphabet},
 $Q$ the \emph{finite set of states}, $q_0 \in Q$
the \emph{start state}, $F \subseteq Q$ the set of \emph{final states} and 
$\delta : Q \times \Sigma \to Q$ is the \emph{totally defined state transition function}.
Here, we do not consider incomplete automata.
%The properties of being deterministic and complete are implied by the definition of $\delta$ as a total function.
The transition function $\delta : Q \times \Sigma \to Q$
extends to a transition function on words $\delta^{\ast} : Q \times \Sigma^{\ast} \to Q$
by setting $\delta^{\ast}(q, \varepsilon) := q$ and $\delta^{\ast}(q, wa) := \delta(\delta^{\ast}(q, w), a)$
for $q \in Q$, $a \in \Sigma$ and $w \in \Sigma^{\ast}$. In the remainder, we drop
the distinction between both functions and also denote this extension by $\delta$.
%If $\delta$ is a partial function, the automaton is called \emph{incomplete}.
%Most of the time, the automata considered in this paper will be complete, deterministic and \emph{initially connected}, the last notion meaning that
%for every $q \in Q$ there exists some $w \in \Sigma^{\ast}$ such that $\delta(q_0, w) = q$.
%A \emph{trap} (\emph{sink}) \emph{state} is a state $q \in Q$
%such that, for any $x \in \Sigma$, $\delta(q, x) = q$.
The language \emph{recognized} by an automaton $\mathcal A = (\Sigma, Q, \delta, q_0, F)$ is
$
 L(\mathcal A) = \{ w \in \Sigma^{\ast} \mid \delta(q_0, w) \in F \}.
$
A language $L \subseteq \Sigma^{\ast}$ is called \emph{regular} if $L = L(\mathcal A)$
for some finite automaton~$\mathcal A$.

The \emph{Nerode right-congruence} 
with respect to $L \subseteq \Sigma^*$ is defined, for $u,v \in \Sigma^*$, by $u \equiv_L v$ if and only if 
$
 \forall x \in \Sigma^* : ux \in L \Leftrightarrow vx \in L.
$
The equivalence class of $w \in \Sigma^{\ast}$
is denoted by $[w]_{\equiv_L} = \{ x \in \Sigma^{\ast} \mid x \equiv_L w \}$.
A language is regular if and only if the above right-congruence has finite index, and it can
be used to define the \emph{minimal deterministic automaton}
$\mathcal A_L = (\Sigma, Q_L, \delta_L, [\varepsilon]_{\equiv_L}, F_L)$
with 
$Q_L  = \{ [u]_{\equiv_L} \mid u \in \Sigma^{\ast} \}$,
$\delta_L([w]_{\equiv_L}, a)  = [wa]_{\equiv_L}$
and $F_L = \{ [u]_{\equiv_L} \mid u \in L \}$.
Let $L \subseteq \Sigma^*$ be regular
 with minimal automaton $\mathcal A_L = (\Sigma, Q_L, \delta_L, [\varepsilon]_{\equiv_L}, F_L)$.
 The number $|Q_L|$ is called the \emph{state complexity} of $L$ and
 denoted by $\stc(L)$. The \emph{state complexity of a regularity-preserving operation}
 on a class of regular languages is the greatest state complexity of the result
 of this operation as a function of the (maximal) state complexities for argument languages from the class.

Given two automata $\mathcal A = (\Sigma, S, \delta, s_0, F)$
and $\mathcal B = (\Sigma, T, \mu, t_0, E)$, an \emph{automaton homomorphism} %nochmal gucken
$h : S \to T$ is a map between the state sets such that 
for each $a \in \Sigma$ and state $s \in S$ we have
$
 h(\delta(s, a)) = \mu(h(s),a),
$
$h(s_0) = t_0$ and $h^{-1}(E) = F$. If $h : S \to T$ is surjective, then $L(\mathcal B) = L(\mathcal A)$. A bijective homomorphism between automata $\mathcal A$
and $\mathcal B$
is called an \emph{isomorphism}, and the two automata are said to be isomorphic.

The \emph{minimal commutative automaton} was introduced in~\cite{GomezA08}
to investigate the learnability of commutative languages. In~\cite{Hoffmann2021NISextended,DBLP:conf/cai/Hoffmann19}
this construction was used to define the index and period vector and in
the derivation of the state complexity bounds mentioned in Table~\ref{tab:sc_known_results}.

\begin{definition}[minimal commutative aut.]
\label{def::min_com_aut}
 Let $L \subseteq \Sigma^*$ be regular. The \emph{minimal commutative  automaton} for $L$
 is $\mathcal C_L = (\Sigma, S_1 \times \ldots \times S_k, \delta, s_0, F)$
 with 
% $S_j := \{ [a_j^m]_{\equiv_L} : m \ge 0 \}$,
% $s_0 := ([\varepsilon]_{\equiv_L}, \ldots, [\varepsilon]_{\equiv_L})$,
% $F  := \{ ([\pi_1(w)]_L, \ldots, [\pi_k(w)]_L) : w \in L \}$
%\begin{align*}
% S_j & := \{ [a_j^m]_{\equiv_L} : m \ge 0 \}, \\
% s_0 & := ([\varepsilon]_{\equiv_L}, \ldots, [\varepsilon]_{\equiv_L}), \\
% F   & := \{ ([\pi_1(w)]_L, \ldots, [\pi_k(w)]_L) : w \in L \}
%\end{align*}
\[
 S_j = \{ [a_j^m]_{\equiv_L} : m \ge 0 \}, \quad
 F  = \{ ([\pi_1(w)]_{\equiv_L}, \ldots, [\pi_k(w)]_{\equiv_L}) : w \in L \}
\]
 and $\delta((s_1, \ldots, s_j, \ldots, s_k), a_j) = (s_1, \ldots, \delta_{j}(s_j, a_j), \ldots, s_k)$
 with one-letter transitions $\delta_{j}([a_j^m]_{\equiv_L}, a_j) = [a_j^{m+1}]_{\equiv_L}$ for $j = 1,\ldots, k$ and $s_0 = ([\varepsilon]_{\equiv_L}, \ldots, [\varepsilon]_{\equiv_L})$. 
\end{definition}
 
 \begin{toappendix}
\begin{remark}
\label{rem:equal_states_C_L}
 Let $L \subseteq \Sigma^*$ be commutative and $\mathcal C_L = (\Sigma, S_1 \times \ldots \times S_k, \delta, s_0, F)$
 be the minimal commutative automaton.
 Note that, by the definition of the transition function in Definition~\ref{def::min_com_aut},
 we have, for all $u,v \in \Sigma^*$,
 \begin{equation}
    \delta(s_0, u) = \delta(s_0, v) \Leftrightarrow \forall j \in \{1,\ldots,k\} : \pi_j(u) \equiv_L \pi_j(v).
 \end{equation}
\end{remark}
\end{toappendix}

\begin{toappendix} 
For a commutative language $L \subseteq \Sigma^*$,
the Nerode right- and left-congruence and the syntactic congruence
coincide. So, as $[u]_{\equiv_L} = [\pi_1(u) \cdots \pi_k(u)]_{\equiv_L}
= [\pi_1(u)]_{\equiv_L} \cdots [\pi_k(u)]_{\equiv_L}$
we have:
\begin{equation}
\label{eqn:proj_equivalent_words_equivalent}
\forall j \in \{1,\ldots,k\} : \pi_j(u) \equiv_L \pi_j(v) 
\Rightarrow u \equiv_L v.
\end{equation}
% With the previous equation, we can deduce, as was done in~\cite{GomezA08},
% the next result.
\end{toappendix}

In~\cite{GomezA08}, the next result was shown.

\begin{theorem}[G{\'{o}}mez \& Alvarez~\cite{GomezA08}]
\label{thm::min_com_aut}
 Let $L \subseteq \Sigma^*$ be a commutative regular language.
 Then, $L = L(\mathcal C_L)$.
 %and $L$ is regular if and only if $\mathcal C_L$ is finite.
\end{theorem} 

In general the minimal commutative automaton is not equal to the 
minimal deterministic and complete automaton for a regular commutative language $L$, see Example~\ref{ex:min_aut_vs_min_com_aut}.

\begin{example} 
\label{ex:min_aut_vs_min_com_aut}
For $L = \{ w \in \Sigma^* \mid |w|_a = 0 \mbox{ or } |w|_b > 0 \}$
with $\Sigma = \{a,b\}$ the minimal deterministic
and complete automaton and the minimal commutative automaton are not the same, see Figure~\ref{fig:ex:min_aut_vs_min_com_aut}. This language is from~\cite{GomezA08}.
In fact, the difference can get quite large, as shown by $L_p = \{ w \in \Sigma^* \mid \sum_{j=1}^k j\cdot |w|_{a_j} \equiv 0 \pmod{p} \}$ for a prime $p > k$. Here, $\stc(L_p) = p$, but $\mathcal C_{L_p}$ has $p^k$ states.

\begin{figure}[tb]
\centering
\begin{minipage}[t]{0.45\textwidth}
\scalebox{.89}{
\begin{tikzpicture}[>=latex',shorten >=1pt,node distance=2cm,on grid,auto]
\node[state, accepting, initial] (q1) {$\varepsilon$};
\node[state, right of=q1] (q2) {$a$};
\node[state, accepting, below of=q1] (q3) {$b$};

\path[->] (q1) edge node {$a$} (q2);
\path[->] (q2) edge [loop right] node {$a$}   (q2);
\path[->] (q3) edge [loop left] node {$a,b$} (q3);
\path[->] (q1) edge node {$b$} (q3);
\path[->] (q2) edge node {$b$} (q3);
\end{tikzpicture}
}
\end{minipage}
\begin{minipage}[t]{0.45\textwidth}
\scalebox{.89}{
\begin{tikzpicture}[>=latex',shorten >=1pt,node distance=2cm,on grid,auto]
\node[state, accepting, initial] (q1) {$\varepsilon,\varepsilon$};
\node[state, right of=q1] (q2) {$a,\varepsilon$};
\node[state, accepting, below of=q1] (q3) {$\varepsilon,b$};
\node[state, accepting, right of=q3] (q4) {$a,b$};

\path[->] (q1) edge node {$a$} (q2);
\path[->] (q2) edge [loop right] node {$a$}   (q2);
\path[->] (q3) edge [loop left] node {$b$} (q3);
\path[->] (q4) edge [loop right] node {$a,b$} (q4);
\path[->] (q1) edge node {$b$} (q3);
\path[->] (q2) edge node {$b$} (q4);
\path[->] (q3) edge node {$a$} (q4);
\end{tikzpicture}
}
\end{minipage}
 \caption{The minimal deterministic automaton (left) and the minimal commutative
  automaton (right) of the language $\{ w \in \Sigma^* \mid |w|_a = 0 \mbox{ or } |w|_b > 0 \}$.}
    \label{fig:ex:min_aut_vs_min_com_aut}
\end{figure}
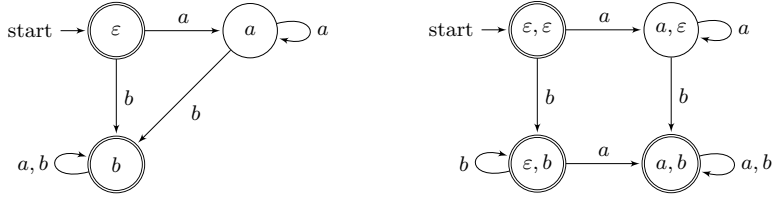
\end{example}

The next definition from~\cite{Hoffmann2021NISextended,DBLP:conf/cai/Hoffmann19} generalizes the notion of a cyclic and non-cyclic part for unary automata
\cite{PighizziniS02},
and the notion of periodic language~\cite{EhrenfeuchtHR83,Hoffmann2021NISextended,DBLP:conf/cai/Hoffmann19}.

\begin{definition}[index and period vector]
\label{def:index_and_period_vector}
  The \emph{index vector} $(i_1, \ldots, i_k)$
  and \emph{period vector} $(p_1, \ldots, p_k)$
  for a commutative regular language $L \subseteq \Sigma^*$ with minimal commutative automaton
  $\mathcal C_L = (\Sigma, S_1 \times \ldots \times S_k, \delta, s_0, F)$
  are the unique minimal numbers such that $\delta(s_0, a_j^{i_j}) = \delta(s_0, a_j^{i_j + p_j})$
  for all $j \in \{1,\ldots,k\}$.
%   Let $L$ be 
%   a commutative regular language with minimal commutative automaton $\mathcal C_L = (\Sigma, S_1 \times \ldots \times S_k, \delta, s_0, F)$.
%   Let $j \in \{1,\ldots,k\}$ and consider the sequence
%   of states $\delta(s_0, a_j^m)$ for $m = 0,1,\ldots$ with respect
%   to the input letter $a_j$.
%   By finiteness there exists 
%   $i_j \ge 0$ and $p_j > 0$ with $i_j + p_j$ minimal such that $\delta(s_0, a_j^{i_j}) = \delta(s_0, a_j^{i_j + p_j})$.
%   As this state sequence could be identified with $S_j$, we have $|S_j| = i_j + p_j$.
%   The vector $(i_1, \ldots, i_k)$ is called the \emph{index vector} and
%   the vector $(p_1, \ldots, p_k)$ the \emph{period vector} of $L$.
\end{definition} %zusammenhang mit periodisch, und wie man periodie draaus berechnet.

Note that, in Definition~\ref{def:index_and_period_vector},
we have, for all $j \in \{1,\ldots,k\}$, $|S_j| = i_j + p_j$. Also note that
for unary languages, i.e., if $|\Sigma| = 1$, $\mathcal C_L$ equals $\mathcal A_L$
and $i_1 + p_1$ equals the number of states of the minimal automaton.

\begin{example}
\label{ex:index_period}
  Let  $L = (aa)^{\ast} \shuffle (bb)^{\ast} \cup (aaaa)^{\ast} \shuffle b^{\ast}$.
 Then $(i_1, i_2) = (0,0)$, $(p_1, p_2) = (4,2)$,
$\pi_1(L) = (a a)^{\ast}$ and $\pi_2(L) = b^{\ast}$.
\end{example}

Let $u, v \in \Sigma^*$.
 Then, $u$ is a \emph{subsequence}\footnote{Also called a \emph{scattered subword}
 in the literature~\cite{DBLP:journals/tcs/GruberHK07,KarandikarNS16}.} of $v$, denoted by $u \preccurlyeq v$,
 if and only if 
 $
   v \in u \shuffle \Sigma^*.
 $
 The thereby given order is called the \emph{subsequence order}.
 Let $L \subseteq \Sigma^*$.
 Then, we define
(1) the \emph{upward closure}
 $\mathop{\uparrow\!} L = L \shuffle \Sigma^* = \{ u \in \Sigma^* : \exists v \in L : v \preccurlyeq u \}$;
(2) the \emph{downward closure} $\mathop{\downarrow\!} L = \{ u \in \Sigma^* : u \shuffle \Sigma^* \cap L \ne \emptyset \} = \{ u \in \Sigma^* : \exists v \in L : u \preccurlyeq v \}$;
(3) the \emph{upward interior}, denoted by $\mathop{\downtouparrow\!} L$,
 as the largest upward-closed set in $L$, i.e. the largest
 subset $U \subseteq L$ such that $\mathop{\uparrow\!} U = U$
 and
(4) the \emph{downward interior}, denoted by $\mathop{\uptodownarrow\!} L$,
 as the largest downward-closed set in $L$, i.e., the largest
 subset $U \subseteq L$ such that $\mathop{\downarrow\!} U = U$.
We have
%\begin{equation}\label{eqn:interior}
$
 \mathop{\uptodownarrow\!} L = \Sigma^* \setminus \mathop{\uparrow\!} (\Sigma^* \setminus L)
$
and
$
 \mathop{\downtouparrow\!} L = \Sigma^* \setminus \mathop{\downarrow\!} (\Sigma^* \setminus L).
$
%see  ~\cite{KarandikarNS16}.

%\end{equation}

The following two results, which will be needed later, are from~\cite{Hoffmann2021NISextended,DBLP:conf/cai/Hoffmann19}.

\begin{theorem}
\label{thm:sc_shuffle}
Let $U,V \subseteq \Sigma^*$ be commutative regular languages with index
and period vectors $(i_1, \ldots, i_k), (j_1, \ldots, j_k)$
and $(p_1, \ldots, p_k), (q_1, \ldots, q_k)$. Then, the index vector of $U \shuffle V$
% for the  index vector $(r_1, \ldots, r_k)$
% and period vector $(s_1, \ldots, s_k)$ of $U\shuffle V$ 
is at most 
\[
 (i_1 + j_1 + \lcm(p_1, q_1) - 1, \ldots, i_k + j_k + \lcm(p_k,q_k) - 1)
\]
and the period vector is at most
$
 (\lcm(p_1, q_1), \ldots, \lcm(p_k, q_k)).
$
So, $\stc(U\shuffle V) \le \prod_{l=1}^k (i_l + j_l + 2\cdot \lcm(p_l, q_l) - 1)$.
\end{theorem}

\begin{theorem}
\label{thm:sc_closure_interior}
 Let $\Sigma = \{a_1, \ldots, a_k\}$.
 Suppose $L \subseteq \Sigma^*$ is commutative and regular
 with index vector $(i_1, \ldots, i_k)$
 and period vector $(p_1, \ldots, p_k)$.
 Then,
 $\max\{\stc(\uparrow\! L),\stc(\downarrow\! L),\stc(\downtouparrow\! L ),\stc(\uptodownarrow\! L) \} \le  \prod_{j=1}^k (i_j + p_j)$.
\end{theorem}

\section{Product-Form Minimal Automata}% and State Complexity}
\label{sec:product-form}

As shown in Example~\ref{ex:min_aut_vs_min_com_aut}, the minimal automaton, in general, 
does not equal 
the minimal commutative automaton. Here, we introduce the class of commutative regular 
languages for which both are isomorphic. The corresponding
commutative languages are called \emph{languages with a minimal automaton
of product-form}, as the minimal commutative automaton is built with the Cartesian product.

\begin{definition}[languages with product-form minimal automaton]
 A commutative and regular language $L \subseteq \Sigma^*$
 is said to have a \emph{minimal automaton of product-form},
 if $\mathcal C_L$ is isomorphic to $\mathcal A_L$.
 %the minimal deterministic and complete automaton is isomorphic to the minimal
 %commutative automaton.
\end{definition}

If $|\Sigma| = 1$, we see easily that $\mathcal C_L$ is the minimal deterministic and complete automaton.

\begin{proposition}
\label{prop:unary_min_prod_form}
 If $|\Sigma| = 1$, then each commutative and regular $L \subseteq \Sigma^*$
 has a minimal automaton of product-form.
 %Every unary language has a minimal automaton of product-form.
 % genauer wenn |Sigma|=1, weil wenn unär übier großrem
 % alphabet nicht.
 More generally, if $L \subseteq\{a\}^*$, then $L \shuffle (\Sigma\setminus\{a\})^*$
 has a minimal automaton of product-form.
\end{proposition}

Apart from the unary languages, we give another example 
of a language with minimal automaton of product-form next.

% \begin{example}
%  Todo.
% % beispiel aus Nis state complexity paper
% % und wetiere

% %
% % lemma um welche zu konstruieren 4.2 aus nis extended paper.
% % in dem beispiel darauf verweisen
% \end{example}

\begin{example}  \label{ex::comm_aut} 
 Let $L = (a a)^{\ast} \shuffle (b b)^{\ast} \cup (a a a)^{\ast} \shuffle b(b b)^{\ast}$ over $\Sigma = \{a,b\}$. 
 See Figure~\ref{fig::example} for the minimal commutative automaton.
 Here, the minimal commutative
 automaton equals the minimal automaton.
%  , for if two states
%  are in the top row and are both final or non-final,
%  we see that after reading $b$ we can distinguish them w.r.t. the Nerode right-congruence by a word from $\{a\}^{\ast}$
%  on the bottom row, and a similar reasoning applies to the bottom row.
%  If we have two states, one from the top row and one from the bottom row,
%  and both final or non-final, we can distinguish them the following way:
%  If the top row state is non-final, reading $a$ or $a a a$
%   maps it to a final state, and one of both words has to map the bottom row state
%   to a non-final state.
%   If the top row state is final, the same argument works, but with final and non-final states
%   switched.
  
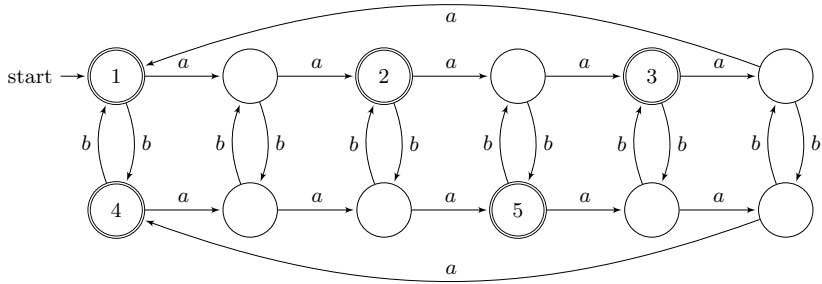
\begin{figure}[htb]
%\begin{figure}[tbh]
\begin{center} 
\scalebox{.89}{
\begin{tikzpicture}[>=latex',shorten >=1pt,node distance=2cm,on grid,auto]
 \node[state, initial, accepting] (1) {1};
 \node[state]                     (2) [right of=1] {};
 \node[state,accepting]           (3) [right of=2] {2};
 \node[state]                     (4) [right of=3] {};
 \node[state,accepting]           (5) [right of=4] {3};
 \node[state]                     (6) [right of=5] {};
 
 \node[state,accepting]           (7)  [below of=1] {4};
 \node[state,]                    (8)  [right of=7] {};
 \node[state,]                    (9)  [right of=8] {};
 \node[state,accepting]           (10) [right of=9] {5};
 \node[state,]                    (11) [right of=10] {};
 \node[state,]                    (12) [right of=11] {};
 
 \path[->] (1) edge node {$a$} (2);
 \path[->] (2) edge node {$a$} (3);
 \path[->] (3) edge node {$a$} (4);
 \path[->] (4) edge node {$a$} (5);
 \path[->] (5) edge node {$a$} (6);
 \path[->] (6) edge [bend right=20] node {$a$} (1);
 
 \path[->] (7) edge node {$a$} (8);
 \path[->] (8) edge node {$a$} (9);
 \path[->] (9) edge node {$a$} (10);
 \path[->] (10) edge node {$a$} (11);
 \path[->] (11) edge node {$a$} (12);
 \path[->] (12) edge [bend left=20,above] node {$a$} (7);
 
 \path[->] (1) edge [bend left=20] node  {$b$} (7);
 \path[->] (7) edge [bend left=20] node  {$b$} (1);
 
 \path[->] (2) edge [bend left=20] node  {$b$} (8);
 \path[->] (8) edge [bend left=20] node  {$b$} (2);
 
 \path[->] (3) edge [bend left=20] node  {$b$} (9);
 \path[->] (9) edge [bend left=20] node  {$b$} (3);
 
 \path[->] (4) edge [bend left=20] node  {$b$} (10);
 \path[->] (10) edge [bend left=20] node  {$b$} (4);
 
 \path[->] (5)  edge [bend left=20] node  {$b$} (11);
 \path[->] (11) edge [bend left=20] node  {$b$} (5);
 
 \path[->] (6)  edge [bend left=20] node  {$b$} (12);
 \path[->] (12) edge [bend left=20] node  {$b$} (6);
\end{tikzpicture}}
\end{center}
 \caption{$\mathcal C_L$ for $L = (a a)^{\ast} \shuffle (b b)^{\ast} \cup (a a a)^{\ast} \shuffle b (b b)^{\ast}$. Here $\mathcal C_L$ is isomorphic to $\mathcal A_L$.} %See Example~\ref{ex::comm_aut}.}
    \label{fig::example}
\end{figure}
\end{example}

However, the next proposition gives a strong necessary criterion
for a commutative language to have a minimal automaton of product-form.\todo{oder product-type?}

\begin{propositionrep}
\label{prop:at_most_one_projection_finite}
 If $L \subseteq \Sigma^*$ is commutative and regular
 with a minimal automaton of product-form,
 then %$|\{ j \in \{1,\ldots, k\} \mid \pi_j(L) \mbox{ is finite} \}| \le 1$.
 $|\{ x \in \Sigma \mid \pi_{\{x\}}(L) \mbox{ is finite } \}| \le 1$.
 So, $\pi_{\Gamma}(L)$ is infinite for $|\Gamma| \ge 2$, in particular no finite language over
 an at least binary alphabet is in this class.
\end{propositionrep}
\begin{proof}
 Suppose we have two distinct $j, j' \in \{1,\ldots,k\}$
 such that $\pi_j(L)$ and $\pi_{j'}(L)$ are finite.
 Set $N = \max\ + 1$.
 Then $a_j^N \equiv_L a_{j'}^N$
 and, as $N > 0$, this implies that the minimal 
 commutative automaton has strictly more states than the minimal deterministic automaton. 
 
 For the last sentence, note that if $a \in \Gamma$, then
 $\pi_{\{a\}}(L) = \pi_{\{a\}}(\pi_{\Gamma}(L))$. Hence, if $\pi_{\Gamma}(L)$
 is finite with $|\Gamma| \ge 2$, then at least two one-letter projection
 languages would be finite too. Hence, with the previous claim, if $L$
 is commutative and regular, it does not have a minimal automaton of product-form in this case. \qed
\end{proof}

For example, $L = \{\varepsilon\}$ over $\Sigma$
does not have a minimal automaton of product-form if $|\Sigma| > 1$.
Recall that the minimal automaton, as defined here, is always complete.
Note that the converse of Proposition~\ref{prop:at_most_one_projection_finite}
is not true, as shown by $aa^*$ over $\Sigma = \{a,b\}$.

In the following statement, we give alternative characterizations
for commutative languages with minimal automata of product-form.

\begin{theoremrep}
\label{thm:product-form_characterizations}
 Let $L \subseteq \Sigma^*$ be a commutative regular language
 with index vector $(i_1, \ldots, i_k)$
 and period vector $(p_1, \ldots, p_k)$.
 The following are equivalent:
 \begin{enumerate}
 \item the minimal automaton has product-form;
 
 \item $\stc(L) = \prod_{j=1}^k (i_j + p_j)$;
 
%  \item the minimal number of generators of the syntactic monoid
%   is $k$;
  
%  \item for any homomorphism $\varphi : \Sigma^* \to M$
%   into a monoid such that $\varphi^{-1}(\varphi(L)) = L$,
%   the monoid $M$ is generated by at least $k$ elements;
  
 \item $u \equiv_L v$ implies $\forall a \in \Sigma : a^{|u|_a} \equiv_L a^{|v|_a}$;
 
 \item $u \equiv_L v$ if and only if $\forall a \in \Sigma : a^{|u|_a} \equiv_L a^{|v|_a}$.
\end{enumerate}
\end{theoremrep}
\begin{proof}
 First, suppose $\mathcal C_L$ is isomorphic to $\mathcal A_L$.
 Then, by the definition of the state complexity, we find $\stc(L) = \prod_{j=1}^k (i_j + p_j)$.
 Conversely, suppose $\stc(L) = \prod_{j=1}^k (i_j + p_j)$.
 As $L$ is commutative, we have $L = L(\mathcal C_L)$.
 Every automaton recognizing $L$ in which all states are reachable
 from the start state can be mapped surjectively onto $\mathcal A_L$, see~\cite{Hopcroft:1971,DBLP:books/daglib/0088160}.
 In particular, this holds true for $\mathcal C_L$.
 By finiteness, as both have the same number of states, this must be an isomorphism.
 Hence, the first two conditions are equivalent
 
 That the last condition is equivalent to the first
 is shown, in a more general context without referring to any previous
 result, in Theorem~\ref{thm:alternative_characterizations},
 as the case of commutative languages corresponds
 to the case $\Sigma_1 = \{a_1\}, \ldots, \Sigma_k = \{a_k\}$
 with the notation from Section~\ref{subsec:subclasses}.\qed
%  By commutativity, the Nerode right-congruence and the syntactic congruence coincide.
%  Also, as $[u]_{\equiv_L} = [a_1^{|u|_{a_1}} \cdots a_k^{|u|_{a_k}}]_{\equiv_L}$,
%  the classes $[a_j]_{\equiv_L}$ generate all Nerode right-congruence classes
%  (and all synctactic congruence classes).
%  Note that $\mathcal C_L$ is defined with powers of these generating classes.
%  Now, suppose we have strictly less than $k$
%  generators. Then, we can write the $[a_j]_{\equiv_L}$ classes
%  in terms of these generators and also the generators in terms of the $[a_j]_{\equiv_L}$.
% 
 % Z_2 x Z_3 vs Z_6, nicht alle kleinsten erzeugendenmengen gleich viele elemente...
 
 That the third condition is equivalent to the last condition 
 follows as for a commutative language we
 have $[u]_{\equiv_L} = [a_1^{|u|_{a_1}} \cdots a_k^{|u|_{a_k}}]_{\equiv_L}$,
 so if $\forall a \in \Sigma : a^{|u|_a} \equiv_L a^{|v|_a}$,
 then $u \equiv_L v$ is always valid for commutative languages.\qed
\end{proof}

Next, we give a way to construct commutative regular languages with minimal automata
of product-form.

\begin{toappendix}
 Let $L \subseteq \Sigma^*$. Next, we will need the following equations:
 \begin{equation}\label{eqn:word_in_shuffle_lang}
  u \in \bigshuffle_{i=1}^k \pi_i(L) \Leftrightarrow \forall i \in \{1,\ldots,k\} : \pi_i(u) \in \pi_i(L),
 \end{equation}
 and
 \begin{equation}\label{eqn:L_in_shuffle_L}
     L \subseteq \bigshuffle_{i=1}^k \pi_i(L).
 \end{equation}

\begin{lemma}%[Decomposition Lemma] 
\label{lem:union_shuffle_projection_lang}
 Let $L \subseteq \Sigma^*$ be a commutative language
 with minimal commutative automaton $\mathcal C_L = (\Sigma, S_1 \times \ldots\times S_k, \delta, s_0, F)$,
 index vector $(i_1, \ldots, i_k)$ and period vector $(p_1, \ldots, p_k)$.
 Set~$n = |F|$. Suppose $F = \{ (s_1^{(l)}, \ldots, s_k^{(l)}) \mid l \in \{ 1,\ldots, n \} \}$
 and set, for $l \in \{1,\ldots, n\}$, $L^{(l)} = \{ w \in \Sigma^* \mid \delta(s_0, w) = (s_1^{(l)}, \ldots, s_k^{(l)}) \}$, 
 and $L_j^{(l)} = \pi_j(L^{(l)})$, $j \in \{1,\ldots,k\}$.
 Then,
 \[
  L = \bigcup_{l = 1}^{n} L^{(l)} \mbox{ and } L^{(l)} = \bigshuffle_{j=1}^k L_j^{(l)}.
 \]
 For $j \in \{1,\ldots,k\}$, set $\mathcal C_{L, \{a_j\}} = (\{a_j\}, S_j, \delta_j, [\varepsilon]_{\equiv_L}, F_j)$
 with
 \[ 
  \delta_j([a_j^m]_{\equiv_L}, a_j)  = [a_j^{m+1}]_{\equiv_L} \mbox{ and }
  F_j  = \{ s_j^{(1)}, \ldots, s_j^{(n)} \}
 \]
 for $[a_j^m]_{\equiv_L} \in S_j$, $m \ge 0$.
 Then, 
 %$U_j^{(l)} = \{ w \in \{a_j\}^* \mid \delta_j([\varepsilon]_{\equiv_l}, w) = s_j^{(l)} \}$
 \[ 
 L_j^{(l)} = L((\{a_j\}, S_j, \delta_j, [\varepsilon]_{\equiv_L}, \{s_j^{(l)}\})
 \mbox{ and }
 L(\mathcal C_{L, \{a_j\}}) = \pi_j(L) = \bigcup_{l=1}^{n} L_j^{(l)}. 
 \]
 for $l \in \{1,\ldots,n\}$ and the automata $\mathcal C_{L,\{a_j\}}$
 have index $i_j$ and period $p_j$.
\end{lemma}
\begin{proof}
 As every accepted word drives $\mathcal C_L$ into some final state,
 we have $L = \bigcup_{l = 1}^{n} L^{(l)}$. Next, we show the other claims.
 
 \begin{claiminproof}
  Let $l \in \{1,\ldots,n\}$. Then $L^{(l)} = \bigshuffle_{j=1}^k L_j^{(l)}$.
 \end{claiminproof}
 \begin{claimproof}
 By Equation~\eqref{eqn:L_in_shuffle_L}, we have $L^{(l)} \subseteq \bigshuffle_{j=1}^k L_j^{(l)}$.
 If $w \in \bigshuffle_{j=1}^k L_j^{(l)}$,
 then, for $j \in \{1,\ldots,k\}$, there exists $u_j \in L_j^{(l)}$
 such that $\pi_j(w) = \pi_j(u_j)$.
 Hence, by Remark~\ref{rem:equal_states_C_L}, $\delta(s_0, w) = \delta(s_0, u_1 \cdots u_j) = (s_1^{(l)}, \ldots, s_k^{(l)})$,
 and so $w \in L^{(l)}$.
 \end{claimproof}
 
 \begin{claiminproof}
  Let $l \in \{1,\ldots,n\}$ and $j \in \{1,\ldots,k\}$. Then 
  \[ L_j^{(l)} = L((\{a_j\}, S_j, \delta_j, [\varepsilon]_{\equiv_L}, \{s_j^{(l)}\}).
  \]
 \end{claiminproof}
 \begin{claimproof}
  By Remark~\ref{rem:equal_states_C_L}, for $u \in a_j^*$, 
 \begin{align*}
     & \delta_j([\varepsilon]_{\equiv_L}, u) = s_j^{(l)} \\
      & \Leftrightarrow [u]_{\equiv_L} = s_j^{(l)} \\
      & \Leftrightarrow \exists w \in L^{(l)} : u \equiv_L \pi_j(w) \\
      & \Leftrightarrow \exists w \in L^{(l)} : \delta(s_0, \pi_1(w)\cdots \pi_{j-1}(w) u \pi_{j+1}(w) \cdots \pi_k(w)) = \delta(s_0, w) \\
      & \Leftrightarrow \exists w \in \Sigma^* : \delta(s_0, \pi_1(w)\cdots \pi_{j-1}(w) u \pi_{j+1}(w) \cdots \pi_k(w)) = (s_1^{(l)}, \ldots, s_k^{(l)}) \\
      & \Leftrightarrow u \in L_j^{(l)}. 
 \end{align*}  
  And this shows the claim.
 \end{claimproof}
 
%  In $\mathcal C_L$, for $u, v \in \Sigma^*$, we have
%  \begin{equation}
%  \label{eqn:C_L_same_state}
%   \delta(s_0, u) = \delta(s_0, v) \Leftrightarrow \forall j \in \{1,\ldots, k\} : \pi_j(u) \equiv_L \pi_j(v).
%  \end{equation}
 As $F_j = \{s_j^{(1)},\ldots, s_j^{(n)}\}$, we
 find $L(\mathcal C_{L,\{a_j\}}) = \bigcup_{l=1}^n L_j^{(l)}$
 and, with the previous equations and $L_j^{(l)} \subseteq a_j^*$,
 \[
  \pi_j(L) = \pi_j\left( \bigcup_{l=1}^{n} \bigshuffle_{i=1}^k L_i^{(l)} \right) 
           = \bigcup_{l=1}^{n}  \pi_j\left(  \bigshuffle_{i=1}^k L_i^{(l)} \right) 
           = \bigcup_{l=1}^{n} L_j^{(l)}.
 \]
 Lastly, that the automata $\mathcal C_{L,\{a_j\}}$ have index $i_j$
 and period $p_j$ is implied as the Nerode right-congruence classes $[a_j^m]_{\equiv_L}$, $m \ge 0$
 are the states of these automata. So, we have shown all equations in the statement.~\qed
\end{proof}

\begin{lemma}
\label{lem::nerode_on_projections}
 Let $L \subseteq \Sigma^*$ be a language with $L = \bigshuffle_{i=1}^k \pi_i(L)$. Then,
 for each $j \in \{1,\ldots k\}$ and $n,m\ge 0$, we have
 $
  a_j^m \equiv_{L} a_j^n
   \Leftrightarrow a_j^m \equiv_{\pi_j(L)} a_j^n,
 $
 where on the right side the equivalence is considered with
 respect to the unary alphabet $\{a_j\}$.
\end{lemma} 
\begin{proof} 
 In the next equations, we will write ``$\Leftrightarrow$'' to mean two formulas are semantically equivalent, and
 ``$\leftrightarrow$'' as an equivalence in the formula itself.
 
 By assumption and Equation~\eqref{eqn:word_in_shuffle_lang}, we have 
 \begin{equation} \label{eqn:nerode_on_projections}
     w \in L \Leftrightarrow \forall r \in \{1,\ldots k \} : \pi_r(w) \in \pi_r(L).
 \end{equation}
% For if $w \in L$, then $w \in \shuffle_{i=1}^k \pi_i(L)$.
% And if there exists, for any $r \in \{1,\ldots, k\}$,
% a word $w_r \in L$ with $\pi_r(w) = \pi_r(w_r)$,
% then $w \in \shuffle_{i=1}^k \pi_i(L)$.

 Now, let $j \in \{1,\ldots k\}$ be fixed and $n,m \ge 0$. We have
 \begin{equation*}
  a_j^m \equiv_L a_j^n  \Leftrightarrow \forall x \in \Sigma^{\ast}  : a_j^m x \in L \leftrightarrow a_j^n x \in L.
 \end{equation*}
 By Equation~\eqref{eqn:nerode_on_projections}, this is equivalent to 
 \begin{multline}\label{eqn:forall_r} 
      \forall x \in \Sigma^{\ast}  : ( \forall r \in \{1,\ldots,k\} : \pi_r(a_j^m x) \in \pi_r(L) ) \\
                           \leftrightarrow ( \forall r \in \{1,\ldots,k\} : \pi_r(a_j^n x) \in \pi_r(L) ).
 \end{multline}
 
 \begin{claiminproof}
  Equation~\eqref{eqn:forall_r} is equivalent to
  \[ 
   \forall x \in \Sigma^{\ast}  : \pi_j(a_j^m x) \in \pi_j(L) \leftrightarrow \pi_j(a_j^n x) \in \pi_j(L).
  \]
 \end{claiminproof}
 \begin{claimproof}
  First, suppose Equation~\eqref{eqn:forall_r} holds true
  and suppose, for $x \in \Sigma^*$,
  we have $\pi_j(a_j^m x) \in \pi_j(L)$.
  Choose, for $r \ne j$, words $u_r \in \pi_r(L)$, which is possible
  as $L = \bigshuffle_{r=1}^k \pi_r(L)$.
  Set $u_j = \pi_j(x)$ and $y = u_1 \cdots u_k$.
  Then, for all $r \in \{1,\ldots,k\}$,
  as $u_r \in a_r^*$, and hence
  \[
   \pi_r(a_j^my) = \left\{ 
   \begin{array}{ll}
    \pi_r(a_j^mx) & \mbox{if } r = j; \\ 
    \pi_r(u_r)    & \mbox{if } r \ne j,
   \end{array}
   \right.
  \]
  we have $\pi_r(a_j^m y) \in \pi_r(L)$.
  So, by Equation~\eqref{eqn:forall_r},
  we can deduce that, for all $r \in \{1,\ldots,k\}$,
  we have $\pi_r(a_j^n y) \in \pi_r(L)$.
  In particular, we find $\pi_j(a_j^n y) \in \pi_j(L)$.
  But $\pi_j(a_j^ny) = \pi_j(a_j^n)\pi_j(y) = \pi_j(a_j^n)\pi_j(x) = \pi_j(a_j^nx)$.
  So, $\pi_j(a_j^nx) \in \pi_j(L)$.
  Similarly, we can show that, for all $x \in \Sigma^*$,
  $\pi_j(a_j^nx) \in \pi_j(L)$ implies $\pi_j(a^mx) \in \pi_j(L)$.

  Conversely, now suppose, for each $x \in \Sigma^*$, we have
  \[
    \pi_j(a_j^m x) \in \pi_j(L) \leftrightarrow \pi_j(a_j^n x) \in \pi_j(L).
  \]
  Let $x \in \Sigma^*$ and suppose, for all $r \in \{1,\ldots, k\}$,
  we have $\pi_r(a_j^mx) \in \pi_r(L)$.
  As, for $r \ne j$, 
  we have $\pi_r(a_j^mx) = \pi_r(x) = \pi_r(a_j^nx)$
  and, as, by assumption, the condition $\pi_j(a_j^mx) \in \pi_j(L)$
  implies $\pi_j(a_j^nx) \in \pi_j(L)$,
  we find that, for all $r \in \{1,\ldots,k\}$,
  we have $\pi_r(a_j^nx) \in \pi_r(L)$.
  Similarly, we can show the other
  implication in Equation~\eqref{eqn:forall_r}.
 \end{claimproof}

 So, by the previous claim, Equation~\eqref{eqn:forall_r} simplifies to
 \begin{align*} 
   & \forall x \in \Sigma^{\ast}  : \pi_j(a_j^m x) \in \pi_j(L) \leftrightarrow \pi_j(a_j^n x) \in \pi_j(L) \\
   & \Leftrightarrow \forall x \in \{a_j\}^{\ast} : a_j^m x \in \pi_j(L) \leftrightarrow a_j^n x \in \pi_j(L) \\ 
   & \Leftrightarrow a_j^m \equiv_{\pi_j(L)} a_j^n,
 \end{align*}
 and we have shown the statement.~\qed
 \end{proof}
\end{toappendix}

\begin{lemmarep}
\label{lem:sc_shuffle_lang}
 Let $\Sigma = \{a_1, \ldots, a_k\}$ and, for $j \in \{1,\ldots,k\}$, $L_j \subseteq \{a_j\}^*$
 be regular and infinite with index $i_j$ and period $p_j$.
 Then,
 $
  \stc\left(\bigshuffle_{j=1}^k L_j\right) = \prod_{j=1}^k\stc(L_j) = \prod_{j=1}^k (i_j + p_j)
 $
 and $\bigshuffle_{j=1}^k L_j$ has index vector $(i_1, \ldots, i_k)$ and period vector $(p_1, \ldots, p_k)$.
 With Thm.~\ref{thm:product-form_characterizations}, $\bigshuffle_{j=1}^k L_j$ has a product-form minimal automaton.
\end{lemmarep}
\begin{proof}
 Set $L = \bigshuffle_{j=1}^k L_j$. Let $u,v \in \Sigma^*$ be two words such that
 \[
  u = a_1^{n_1} \cdot\ldots\cdot a_k^{n_k} \mbox{ and }
  v = a_1^{m_1} \cdot\ldots\cdot a_k^{n_k}.
 \]
 with $0\le n_j, m_j < \stc(L_j)$, $j \in \{1,\ldots,k\}$.
 Suppose there exists $r \in \{1,\ldots, k\}$
 such that $n_r \ne m_r$. As $L_r$ is unary and $\max\{n_r,m_r\}<\stc(L_r)$, 
 \[ \varepsilon, a_r, a_r^2, \ldots, a_r^{\max\{n_r,m_r\}-1} \]
 are represenatives of distinct Nerode right-equivalence classes of $L_r$.
 Hence, there exists $l_r > 0$ such that, without loss of generality,
 %\footnote{Actually, for infinite unary
 %languages, we can select the state that is mapped to a final state. As, suppose $s \in F$
 %and $t \notin F$ are two states on the unique cycle of the minimal automaton
 %of a unary language and $\delta(s, a^n) = t$.
 %Then, there exists a minimal $r > 1$
 %such that $\delta(s, a^{rn}) \notin F$ and $\delta(s, a^{rn}) \in F$,
 %as $\delta(s, a^n) \notin F$ and there exists a smallest $R > 0$
 %such that $\delta(s, a^{Rn}) \in F$. So, for some $0 < r < R$
 %we find $\delta(s, a^{rn}) \notin F$ and $\delta(s, a^{(r+1)}n) \in F$.
 %But $\delta(s, a^{(r+1)n}) = \delta(t, a^{rn})$.}, 
 \[
  a_r^{n_r + l_r} \in L_r \mbox{ and } a_r^{m_r + l_r} \notin L_r.
 \]
 As all the $L_j$, $j \in \{1,\ldots, k\}$, are infinite, there exists
 $l_j \ge 0$ for $j \in \{1,\ldots, k\}\setminus\{r\}$
 such that $a_j^{n_j + l_j} \in L_j$.
 Then,
 \[
  a_1^{n_1 + l_1} \cdot\ldots\cdot a_k^{n_k + l_k} \in L.
 \] 
 And, as  $a_r^{m_r + l_r} \notin L_r$, we find
 \[
   a_1^{m_1 + l_1} \cdot\ldots\cdot a_k^{m_k + l_k} \notin L.
 \] 
 Set $x = a_1^{l_1}\cdot\ldots\cdot a_k^{l_k}$. 
 Then, as $L$ is a commutative language,
 \[
  ux \in L \mbox{ and } vx \notin L,
 \]
 i.e., $u \not\equiv_L v$.
 Hence, all words of the form $a_1^{n_1} \cdot\ldots\cdot a_k^{n_k}$
 with $0 \le n_j < \stc(L_j)$ are pairwise non-equivalent.
 So, $\prod_{j=1}^k\stc(L_j) \le \stc(L)$.
 Let $j \in \{1,\ldots,k\}$ and $\mathcal A_{L_j} = (\{a_j\}, Q_{L_j}, \delta_{L_j}, [\varepsilon]_{\equiv_{L_j}}, F_{L_j})$
 be the minimal automaton
 for $L_j$. Also, let $\mathcal C_L = (\Sigma, S_1 \times \ldots\times S_k, \delta, s_0, F)$
 be the minimal commutative automaton from Definition~\ref{def::min_com_aut}
 and consider the automaton $\mathcal C_{L, \{a_j\}} = (\{a_j\}, S_j, \delta_j, [\varepsilon]_{\equiv_L}, F_j)$ 
 from Lemma~\ref{lem:union_shuffle_projection_lang}.
 By Theorem~\ref{thm::min_com_aut}, we have $L = L(\mathcal C_L)$. 
 Hence, $\stc(L) \le \prod_{j=1}^k |S_j|$.
 By Lemma~\ref{lem::nerode_on_projections}, as $L_j = \pi_j(L)$,
 the map
 \[
  [a_j^n]_{\equiv_L} \mapsto [a_j^n]_{\equiv_{L_j}}
 \]
 is a well-defined isomorphism between $\mathcal A_{L_j}$
 and $\mathcal C_{L, \{a_j\}}$. So, $|S_j| = \stc(L_j)$.
 Hence,
 \[
  \stc(L) \le \prod_{j=1}^k \stc(L_j),
 \]
 and, combining everything, $\stc(L) = \prod_{j=1}^k \stc(L_j)$.
 The entry of the index and period vector of $L$ for $j \in \{1,\ldots,k\}$
 is the index and period of the unary automaton $\mathcal C_{L,\{a_j\}}$.
 As $\mathcal C_{L,\{a_j\}}$ and $\mathcal A_{L_j}$ are isomorphic,
 they have the same index and period. So, all claims of the 
 statement are shown.~\qed
\end{proof} 

 In the next theorem and the following remark, we investigate closure properties of the class in question.

\begin{theoremrep}
\label{thm:closure_properties}
 The class of commutative regular languages with minimal automata
 of product-form is closed under left and right quotients
 and complementation. It is not closed under union, intersection and projection.
\end{theoremrep}
\begin{proof} 
 As for every language $L \subseteq \Sigma^*$,
 the state complexity of $L$ and the complement of $L$
 are equal (recall we are only concerned with complete automata here)
 and $\mathcal A_L$ is minimal for both languages, closure
 of the class under complementation follows.
 
 For commutative languages, left and right quotients
 give the same sets, i.e., $u^{-1} L = Lu^{-1}$,
 hence it is sufficient to show closure under left quotients.
 
 Now, if $\mathcal A_L = (\Sigma, Q_L, \delta_L, [\varepsilon]_{\equiv_L}, F_L)$
 is the minimal automaton, then, for $u \in \Sigma^*$,
 $\mathcal B = (\Sigma, Q_L, \delta_L, [u]_{\equiv_L}, F_L)$
 recognizes $u^{-1}L$. In fact, as for every $[u]_{\equiv_L} \ne [v]_{\equiv_L}$,
 there exists $x \in \Sigma^*$ such that $ux \in L$ and $vx \notin L$, or vice versa.
 So, we have $[ux]_{\equiv_L} \in F_L$ and $[vx]_{\equiv_L}\notin F_L$, or vice versa.
 This implies, after discarding states not reachable from the start state,
 that $\mathcal B$ is isomorphic to the minimal automaton of $u^{-1}L$.
 Now, observe that $[w]_{\equiv_L}$ in $\mathcal A_L$
 corresponds to the state $([a_1^{|w|_{a_1}}]_{\equiv_L},\ldots,[a_k^{|w|_{a_k}}]_{\equiv_L})$
 in $\mathcal C_L$.\todo{das genauer?}
 So, $\mathcal B$ is of product-form, which gives the claim.
 Note that, for every $j \in \{1,\ldots,k\}$,
 the index for the letter $a_j$ of $u^{-1}L$
 is $\max\{ i_j - |u|_{a_j}, 0 \}$
 and the period for the letter $a_j$ is the same.
 
 See Remark~\ref{rem:union_intersection_not_closed}
 for examples of languages which show that the class considered is not closed
 under union, intersection and projection.\qed
\end{proof} 

\begin{remark}
\label{rem:union_intersection_not_closed} %beispiel + demorgan
 %Let $\Sigma = \{a,b\}$. % für binäres alphabet nicht, sosnt natürlich immer mti prod 1,
 % conclusion learnability + untere schranke offenes problem
 %Consider $U = \{\varepsilon\} \cup bb^* \shuffle (aa)^*$
 %and $V = \{\varepsilon\} \cup bb^* \shuffle a(aa)^*$.
 %Both languages have minimal automata of product-form.
 %However, $U \cup V = \{\varepsilon\} \cup bb^* \shuffle a^*$
 %is the language from Example~\ref{ex:min_aut_vs_min_com_aut}
 %whose minimal automaton is not in product-form.
 We have $a\shuffle b^* \cap a^* \shuffle b = a\shuffle b$,
 showing, using Proposition~\ref{prop:unary_min_prod_form} and~\ref{prop:at_most_one_projection_finite},
 that this class is not closed under intersection
 and by DeMorgan's laws, as we have closure under complementation, 
 we also cannot have closure under union. 
 Also, $L = aa^* \shuffle bb^* \shuffle cc^* \cup bb^* \shuffle a^* \cup b^*$ \todo{beweisen}
 has a minimal automaton of product-form, but $\pi_{\{a,b\}}(L) = bb^* \shuffle a^* \cup b^*$
 is the language from Example~\ref{ex:min_aut_vs_min_com_aut}. So, this class
 is also not closed under projection.
\end{remark}

\begin{theoremrep}
\label{thm:sc_results}
% endlich udn eine unendlich
% ansonsten erstmal nur nm
%
% a^{p-1}(a^p)* cup a^p shuffle
 %Suppose $U, V \subseteq \Sigma^*$ are commutative and
 %regular with product-form minimal automata
 %of size $n$ and $m$ respectively.
 %Then,
 Let $U, V \subseteq \Sigma^*$ be commutative regular languages
 with product-form minimal automata with $\stc(U) = n$
 and $\stc(V) = m$.

 \begin{enumerate}
 \item \label{thm:sc:shuffle} 
  We have $\stc(U \shuffle V) \le 2nm$ if $|\Sigma| > 1$
 and $\stc(U \shuffle V) \le nm$ if $|\Sigma| = 1$.
 Furthermore, for any $\Sigma$, there exist
 $U, V$ as above such that $nm \le \stc(U \shuffle V)$.
 
 \item In the worst case, $n$ states are sufficient
  and necessary for a DFA to recognize~$\uparrow\! U$.
  Similarly for the
  downward closure and interior
  operations.
 \item In the worst case, $n$ states are sufficient
  and necessary for a DFA to recognize
  the projection of~$U$.
  
 \item In the worst case, $nm$ states are sufficient
  and necessary for a DFA to recognize
  $U \cap V$ or $U \cup V$.
 \end{enumerate}
\end{theoremrep}
\begin{proof}
\begin{enumerate}
\item If $|\Sigma| = 1$, then the shuffle operation
 is the same as concatenation and the claim follows
 by the state complexity of concatenation
 on unary languages, see~\cite{YuZhuangSalomaa1994}.
 Otherwise, with Theorem~\ref{thm:sc_shuffle},
 \begin{multline*} 
  \stc(U \shuffle V) \le \prod_{l=1}^k (i_l + j_l + 2\lcm(p_l, q_l) - 1) 
   \\ \le \prod_{l=1}^k 2(i_l + p_l)(j_l + q_l) 
    = 2 \cdot \prod_{l=1}^k (i_l + p_l) \prod_{l=1}^k (j_l + q_l) 
    = 2 n m.
 \end{multline*}

 Let $p, q > 0$ be two coprime numbers
 and set $U = \bigshuffle_{j=1}^k a_j^{p-1}(a_j^{p})^*$
 and $V = \bigshuffle_{j=1}^k a_j^{q-1}(a_j^{q})^*$.
 By Lemma~\ref{lem:sc_shuffle_lang}
 both have minimal automata of product-form
 and $\stc(U) = p^k$ and $\stc(V) = q^k$
 Also, $\stc( a_j^{p-1}(a_j^{p})^* \cdot a_j^{q-1}(a_j^{q})^* ) = pq$, see~\cite[Lemma 5.1 and Fact 5.2]{YuZhuangSalomaa1994}.
 Further,
 \[
  U \shuffle V = \bigshuffle_{j=1}^k ( a_j^{p-1}(a_j^{p})^* \cdot a_j^{q-1}(a_j^{q})^* ).
 \]
 So, with Lemma~\ref{lem:sc_shuffle_lang},
 the minimal automaton of $U \shuffle V$
 has product-form and $\stc(U\shuffle V) = p^k \cdot q^k$.
 
\item This is a direct application
 of Theorem~\ref{thm:sc_closure_interior},
 as by assumption
 and Theorem~\ref{thm:product-form_characterizations}
 we have $n = \prod_{l=1}^k (i_l + p_l)$.
 For the lower bound, use $L = a^n \shuffle (\Sigma\setminus\{a\})^*$.
 Then, the upward closure is $a^na^* \shuffle (\Sigma\setminus\{a\})^*$
 and the downward closure is $\{\varepsilon,a,\ldots,a^n\}\shuffle (\Sigma\setminus\{a\})^*$.
 
\item Set $L = \shuffle_{j=1}^k a_j^* = \Sigma^*$.
 Then, $\stc(L) = 1$. Or, for any $n > 1$, where we assume $\Sigma = \{a,b\}$
 for notational simplicity, we have $L = a^* \shuffle b^nb^*$, then $\stc(L) = n+1$
 and $\stc(\pi_{\{b\}}(L)) = n+1$.

\item The stated bound is valid for union and intersection
 in general, see~\cite[Theorem 4.3]{YuZhuangSalomaa1994}.
 
 Let $p, q > 0$ be two coprime numbers
 and set $U = \bigshuffle_{j=1}^k a_j^{p-1}(a_j^{p})^*$
 and $V = \bigshuffle_{j=1}^k a_j^{q-1}(a_j^{q})^*$.
 By Lemma~\ref{lem:sc_shuffle_lang}
 both have minimal automata of product-form.
 Also, $\stc( a_j^{p-1}(a_j^{p})^* \cap a_j^{q-1}(a_j^{q})^* ) = pq$
 and
 \[
  U \cap V = \bigshuffle_{j=1}^k ( a_j^{p-1}(a_j^{p})^* \cap a_j^{q-1}(a_j^{q})^* ).
 \]
 So, with Lemma~\ref{lem:sc_shuffle_lang},
 the minimal automaton of $U \cap V$
 has product-form. As this property is closed
 under complementation by Theorem~\ref{thm:closure_properties}
 the statement for union is implied. 
\end{enumerate}
 This finishes the proof.\qed
\end{proof}

\begin{remark} % rmekar U ohne 4-b-zyklus noch kleiner
\label{rem:lower_bound_shuffle}
 I do not know if the bound $2nm$ stated in Theorem~\ref{thm:sc:shuffle}
 for the shuffle operation is tight, but the next example
 shows that if we have a binary alphabet, we can find commutative languages
 with state complexities $n$ and $m$
 and product-form minimal automata whose shuffle
 needs an automaton with strictly more than $nm$ states.
 A similar construction works for more than two letters.
Let $p, q > 11$ be two coprime numbers.
Set $U = a \shuffle b^{p-1}(b^p)^* \cup a^{p-1}(a^p)^* \shuffle bb^{p-1}(b^p)^*$
and $V = b^{q-1}(b^q)^* \cup a^{q-1}(a^q)^*  \shuffle bbb^{q-1}(b^q)^*$.
Then, using that shuffle distributes over union
and a number-theoretical result from~\cite[Lemma 5.1]{YuZhuangSalomaa1994},
we find
\begin{multline*}
    U \shuffle V = a \shuffle W \cup a^{p-1} (a^p)^* \shuffle bW \cup \\ 
            a^q(a^q)^* \shuffle bbW \cup a^{q-1 + p - 1}(a^p)^* (a^q)^* \shuffle bbbW,
\end{multline*} 
where $a^{q-1 + p - 1}(a^p)^* (a^q)^* = F \cup a^{pq - 1}a^*$
for some finite set $F \subseteq \{\varepsilon, a, \ldots, a^{pq - 3} \}$
and $W = E \cup b^{pq-1}b^*$ for some $E \subseteq \{\varepsilon, b, \ldots, b^{pq-3} \}$.
Note that by~\cite[Lemma 5.1]{YuZhuangSalomaa1994} we have $a^{pq-2} \shuffle bbbW \cap U \shuffle V = \emptyset$. % auf den gehen, um zu zeigen sind unterscheidbar.
All languages involved have a product-form minimal automaton. % beweisen!
The minimal automaton for $U$ has $(2 + p) \cdot (1+p)$
states, the minimal automaton for $V$ has $(1 + q)\cdot (q+2)$ states
and that for $U \shuffle V$ has $2pq\cdot (pq+3)$ states.
As $(p-11)(q-11) > 0$ we
can deduce $(1+p)(2+p)(1+q)(2+q) < 2(pq)^2 < 2pq(pq+3)$.
\end{remark}

% \begin{remark}\todo{schon mti proposition gesagt, aber als nciht-beispiel behalten?}
%   In Lemma~\ref{lem:sc_shuffle_lang}, the assumption that
%   the languages $L_j$, with $j \in \{1,\ldots,k\}$, are infinite is essential.
%   For example, for $L = \{a\}\shuffle \{b\}$
%   we find $\stc(L) = 5$, but, as unary languages over $\{a\}$
%   and $\{b\}$ respectively, $\stc(\{a\}) = \stc(\{b\}) = 3$.
% \end{remark}

\section{Partial Commutativity and Other Subclasses}

\label{subsec:subclasses}

% relation to (mar) trace languages
A \emph{partial commutation} on $\Sigma$ is a symmetric and irreflexive relation $I \subseteq \Sigma \times \Sigma$, often called
the \emph{independence relation}. Of interest is the \emph{congruence} $\sim_I$ generated on $\Sigma^*$
by the relation
$
 \{ (ab, ba) \mid (a,b) \in I \}. 
$
A language $L \subseteq \Sigma^*$ is \emph{closed under $I$-commutation}
if $u \in L$ and $u \sim_I v$ implies $v \in L$.
If $I = \{ (a,b) \in \Sigma \times \Sigma \mid a \ne b \}$,
then the languages closed under $I$-commutation are precisely the commutative languages.

Languages closed under some partial commutation relation have been extensively studied,
see~\cite{DBLP:journals/iandc/GomezGP13}, also for further references,
and in particular with relation to (Mazurkiewicz) trace theory~\cite{DBLP:books/ws/95/DR1995,DBLP:journals/iandc/GomezGP13,mazurkiewicz77}, a formalism to describe the execution histories
of concurrent programs.

Here, we will focus on the case that $(\Sigma\times \Sigma) \setminus I$ is transitive, i.e.,
if $u \not\sim_I v$ and $v \not\sim_I w$ implies $u \not\sim_I w$.
In this case, $(\Sigma\times \Sigma) \setminus I$ is an equivalence relation
and we will write $\Sigma_1, \ldots, \Sigma_k$ for the different equivalence classes.

The reason to focus on this particular generalization is, as we will see later, that the definition of the minimal commutative
automaton transfers to this more general setting without much difficulty.

To ease the notation, if we have a partial commutation relation as above with a corresponding
partition $\Sigma = \Sigma_1 \cup \ldots \Sigma_k$ of the alphabet,
we also write $\mathcal L_{\Sigma_1, \ldots, \Sigma_k}$ for the \emph{class of languages
closed under this partial commutation}. Then, as is easily seen,
we have $L \in \mathcal L_{\Sigma_1, \ldots, \Sigma_k}$
if and only if, for $x \in \Sigma_i$, $y \in \Sigma_j$ ($i \ne j$)
and each $u, v \in \Sigma^*$ we have
$
 uxyv \in L \Leftrightarrow uyxv \in L.
$
For example, $L$ is commutative
if and only if $L \in \mathcal L_{\{a_1\}, \ldots, \{a_k\}}$
for $\Sigma = \{a_1, \ldots, a_k\}$.

\begin{toappendix}
\begin{example}
\label{ex:commutate}
 Let $\Sigma = \Sigma_1 \cup \Sigma_2$
 with $\Sigma_1 = \{a,b\}$ and $\Sigma_2 = \{c,d\}$.
 Then, in the language $U = \{abcd, acbd, cabd, cadb, cdab \}$
 the letters commute according to the partition $\Sigma = \Sigma_1 \cup \Sigma_2$.
 On the contrary, in the language $V = \{abc, bac, cba\}$
 the letters do not commute according to the partition,
 as, for example, $abc \in L$, but $cab \notin L$.
\end{example}

\begin{remark}
\label{rem:letters_commutate_with_decomposition}
 Let $\Sigma = \Sigma_1 \cup \ldots \cup \Sigma_k$
 and $L \subseteq \Sigma^*$ be
 such that the letters in $L$
 commute according to the partition.
 Then, for each $w \in \Sigma^*$,
 \[
   w \in L \Leftrightarrow 
  \pi_{\Sigma_1}(w) \pi_{\Sigma_2}(w) \cdot \ldots \cdot \pi_{\Sigma_k}(w) \in L
 \]
 and
 $
  w \equiv_L \pi_{\Sigma_1}(w) \pi_{\Sigma_2}(w) \cdot \ldots \cdot \pi_{\Sigma_k}(w).
 $
\end{remark}
\end{toappendix}

\begin{toappendix}

The next lemma will be needed when we relate the different subclasses
in Subsection~\ref{subsec:subclasses}.

\begin{lemmarep}
\label{lem:equivalence_classes_projection}
 If $L \in \mathcal L_{\Sigma_1, \ldots, \Sigma_k}$
 and $u \in \Sigma_i^*$,
 then $[u]_{\equiv_L} \cap \Sigma_i^* \subseteq [u]_{\equiv_{\pi_{\Sigma_i}(L)}} \cap \Sigma_i^*$.
\end{lemmarep}
\begin{proof}
 Suppose $u,v \in \Sigma_i^*$, $i \in \{1,\ldots,k\}$, and $u \equiv_L v$.
     If $x \in \Sigma_i^*$ and $ux \in \pi_{\Sigma_i}(L)$,
     then $ux = \pi_{\Sigma_i}(w)$ for some $w \in L$.
     Set 
     \[ 
      w' = \pi_{\Sigma_1}(w) \cdots \pi_{\Sigma_{i-1}}(w)vx\pi_{\Sigma_{i+1}}(w)\cdots \pi_{\Sigma_k}(w).
     \]
     Then, using that $\equiv_L$ is a right-congruence, $L \in \mathcal L_{\Sigma_1, \ldots, \Sigma_k}$
     and Remark~\ref{rem:letters_commutate_with_decomposition},
     \begin{align*} 
      [w]_{\equiv_L}
       & = [\pi_{\Sigma_1}(w) \cdots \pi_{\Sigma_{i-1}}(w)ux\pi_{\Sigma_{i+1}}(w)\cdots \pi_{\Sigma_k}(w)]_{\equiv_L} \\
       & = [u \pi_{\Sigma_1}(w) \cdots \pi_{\Sigma_{i-1}}(w)x\pi_{\Sigma_{i+1}}(w)\cdots \pi_{\Sigma_k}(w)]_{\equiv_L} \\
       & = [v \pi_{\Sigma_1}(w) \cdots \pi_{\Sigma_{i-1}}(w)x\pi_{\Sigma_{i+1}}(w)\cdots \pi_{\Sigma_k}(w)]_{\equiv_L} \\
       & = [\pi_{\Sigma_1}(w) \cdots \pi_{\Sigma_{i-1}}(w)vx\pi_{\Sigma_{i+1}}(w)\cdots \pi_{\Sigma_k}(w)]_{\equiv_L} 
     \end{align*}
     So, $w' \equiv_L w$, and hence, $w' \in L$. Also, $\pi_{\Sigma_i}(w') = vx$.
     So, $vx \in \pi_{\Sigma_i}(L)$.~\qed
\end{proof}

\end{toappendix}

\subsection{The Canonical Automaton}
\label{subsec:can_aut}

Here, we generalize our notion of commutative minimal automaton, Definition~\ref{def::min_com_aut},
to have uniform recognition devices for languages in $\mathcal L_{\Sigma_1,\ldots,\Sigma_k}$.

\begin{definition}
\label{def:generalization_canonical_aut}
 Let $\Sigma = \Sigma_1 \cup \ldots \cup \Sigma_k$ be a partition
 and $L \subseteq \Sigma^*$.
 Set $\mathcal C_{L, \Sigma_1, \ldots, \Sigma_k} = (\Sigma, S_1 \times \ldots \times S_k,
 \delta, s_0, F)$
 with, for $i \in \{1,\ldots, k\}$,
 %\begin{align*} 
 $
  S_i  = \{ [u]_{\equiv_L} \mid u \in \Sigma_i^* \}$,
  $F    = \{ ([\pi_{\Sigma_1}(u)]_{\equiv_L}, \ldots, [\pi_{\Sigma_k}(u)]_{\equiv_L}) \mid u \in L \}$,
 %\end{align*}
 $s_0 = ( [\varepsilon]_{\equiv_L}, \ldots, [\varepsilon]_{\equiv_L})$
 and, for $x \in \Sigma_i$,
 \[
 \delta(([u_1]_{\equiv_L}, \ldots, [u_i]_{\equiv_L}, \ldots, [u_k]_{\equiv_L}), x)
   = ([u_1]_{\equiv_L}, \ldots, [u_ix]_{\equiv_L}, \ldots, [u_k]_{\equiv_L})
 \]
 with words $u_j \in \Sigma_j^*$, $j \in \{1,\ldots, k\}$.
 This is called the \emph{canonical automaton} for the given $L$ with
 respect to $\Sigma = \Sigma_1 \cup \ldots \cup \Sigma_k$.
\end{definition}

\begin{toappendix}
\begin{remark}
\label{rem:can_aut_transition}
 Let $\Sigma = \Sigma_1 \cup \ldots \cup \Sigma_k$ be a partition, $u_i,v_i \in \Sigma^*$, $i \in \{1,\ldots,k\}$,
 and $w \in \Sigma^*$.
 Then, for the canonical automaton $\mathcal C_{L, \Sigma_1, \ldots, \Sigma_k}$, 
 we have
 \begin{multline*}
  \delta(([u_1]_{\equiv_L}, \ldots, [u_k]_{\equiv_L}), w)
   = ([v_1]_{\equiv_L}, \ldots, [v_k]_{\equiv_L}) \\
   \Leftrightarrow 
   \forall i \in \{1,\ldots k\} : u_i\pi_{\Sigma_i}(w) \equiv_L v_i.
 \end{multline*}
\end{remark}
\end{toappendix}

\begin{toappendix} 
% für einen aus L_{sigma_1,...} und einen, für den das nicht gilt. findet sich aber auch in den sonstigen beispielen.
\begin{example}
 \begin{figure}[htb]
     \centering
    \scalebox{.75}{    
 \begin{tikzpicture}[>=latex',shorten >=1pt,node distance=4cm,on grid,auto]
  \node[state, initial, accepting] (1) {};
  \node[state]                     (2) [right of=1] {};
  \node[state]                     (3) [below of=2] {};
  \node[state]                     (4) [left of=3]  {};
  
  \path[->] (2) edge [loop right] node {$c$} (2)
            (4) edge [loop left] node {$c$} (4);
  
  \path[->] (1) edge node {$a$} (2)
            (2) edge node {$a$} (3)
            (3) edge node {$a$} (4)
            (4) edge node {$a$} (1);
            
 \path[->] (1) edge [bend left=10,pos=0.2] node {$b,c$} (3)
           (3) edge [bend left=10,pos=0.2] node {$b,c$} (1);
           
 \path[->] (2) edge [bend left=10,pos=0.2] node {$b$} (4)
           (4) edge [bend left=10,pos=0.2] node {$b$} (2);
 
 \end{tikzpicture}}
  % https://tex.stackexchange.com/questions/115145/how-to-define-the-default-vertical-distance-between-nodes
 \scalebox{.75}{   
 \begin{tikzpicture}[>=latex',shorten >=1pt,node distance=2cm and 4cm,on grid,auto]
  \node[state, initial, accepting] (1) {};
  \node[state]                     (2) [below= of 1] {};
  \node[state]                     (3) [below= of 2] {};
  \node[state]                     (4) [below= of 3] {};
  
  \node[state]                     (b1) [right= of 1]  {};
  \node[state]                     (b2) [below= of b1] {};
  \node[state]                     (b3) [below= of b2] {};
  \node[state]                     (b4) [below= of b3] {};
  
  \path[->] (2) edge [loop left] node {$c$} (2)
            (4) edge [loop left] node {$c$} (4);
  \path[->] (b2) edge [loop right] node {$c$} (b2)
            (b4) edge [loop right] node {$c$} (b4);

  \path[->] (1) edge node {$a$} (2)
            (2) edge node {$a$} (3)
            (3) edge node {$a$} (4)
            (4) edge [bend left=55] node {$a$} (1);
            
 \path[->] (1) edge [bend left=40,pos=0.7] node {$c$} (3)
           (3) edge [bend left=55,pos=0.3] node {$c$} (1);
           
 %\path[->] (2) edge [bend left=10,pos=0.2] node {$b$} (4)
 %          (4) edge [bend left=10,pos=0.2] node {$b$} (2);

  \path[->] (b1) edge node {$a$} (b2)
            (b2) edge node {$a$} (b3)
            (b3) edge node {$a$} (b4)
            (b4) edge [bend right=55,right]  node {$a$} (b1);
  \path[->] (b1) edge [bend left=25,pos=0.7] node {$c$} (b3)
            (b3) edge [bend left=55,pos=0.3] node {$c$} (b1);
        
  \path[->] (1) edge [bend left=10] node {$b$} (b1)
            (b1) edge[bend left=10] node {$b$} (1);
  \path[->] (2) edge [bend left=10] node {$b$} (b2)
            (b2) edge[bend left=10] node {$b$} (2);
  \path[->] (3) edge [bend left=10] node {$b$} (b3)
            (b3) edge[bend left=10] node {$b$} (3);
  \path[->] (4) edge [bend left=10] node {$b$} (b4)
            (b4) edge[bend left=10] node {$b$} (4);
 \end{tikzpicture}}
 
  \caption{On the left an automaton recognizing a language $L \in \mathcal L_{\{a,c\}, \{b\}}$
  and on the right the automaton $\mathcal C_{L, \{a,c,\}, \{b\}}$.}
  \label{fig:ex_gen_can_aut}
\end{figure}

 Please see Figure~\ref{fig:ex_gen_can_aut}
 for an automaton that recognizes a language in $\mathcal L_{\{a,c\},\{b\}}$
 and the canonical automaton derived from it.
 That the recognized language is indeed in $\mathcal L_{\{a,c\},\{b\}}$
 can be seen by noting that for each state $q \in Q$ 
 we have $\delta(q, ab) = \delta(q, ba)$
 and $\delta(q, cb) = \delta(q, bc)$.
\end{example}

\begin{example} 
 Let $\Sigma = \{a,b,c\}$ and $L = \Sigma^* ac^*b \Sigma^*$.
 Then $L \in \mathcal L_{\{a,b\},\{c\}}$. The minimal automaton of $L$
 has a single final state and is isomorphic to $\mathcal C_{L, \{a,b\},\{c\}}$,
 as $[\varepsilon]_{\equiv_L} = [c^n]_{\equiv_L}$ for each $n \ge 0$.
 The language $L \cap \Sigma^* c \Sigma^*$
 also has the property that the minimal automaton is isomorphic
 to $\mathcal C_{L \cap \Sigma^* c \Sigma^*, \{a,b\},\{c\}}$.
\end{example}
\end{toappendix}

Next, we show that the canonical automata recognize precisely the
languages in $\mathcal L_{\Sigma_1, \ldots, \Sigma_k}$.
Note that we have dropped the assumption of regularity of $L$.

\begin{theoremrep} 
\label{thm:canonical_aut}
 Let $L \subseteq \Sigma^*$ and $\Sigma = \Sigma_1 \cup \ldots \cup \Sigma_k$ be a partition.
 Then, 
 \begin{enumerate} 
 \item $L \subseteq L(\mathcal C_{L,\Sigma_1,\ldots,\Sigma_k})$ and $L(\mathcal C_{L,\Sigma_1,\ldots,\Sigma_k}) \in \mathcal L_{\Sigma_1, \ldots, \Sigma_k}$.
 \item $L = L(\mathcal C_{L,\Sigma_1,\ldots,\Sigma_k}) \Leftrightarrow L \in \mathcal L_{\Sigma_1, \ldots, \Sigma_k}$.
 \item Let $L \in \mathcal L_{\Sigma_1,\ldots,\Sigma_k}$. Then
  $L$ is regular if and only if $\mathcal C_{L,\Sigma_1,\ldots,\Sigma_k}$ 
  is finite. 
 \end{enumerate}
\end{theoremrep}
\begin{proof} 
\begin{enumerate}
 \item Let $\mathcal C_{L, \Sigma_1, \ldots, \Sigma_k} = (\Sigma, S_1 \times \ldots \times S_k,
 \delta, s_0, F)$ be the canonical automaton for $L$.
 For each $w \in \Sigma^*$, we have
 \[
  \delta(s_0, w)
   = ([\pi_{\Sigma_1}(w)]_{\equiv_L}, \ldots, [\pi_{\Sigma_k}(w)]_{\equiv_L}).
 \]
 If $w \in L$, then, by Definition~\ref{def:generalization_canonical_aut},
 we find $\delta(s_0, w) \in F$.
 So, $L \subseteq L(\mathcal C_{L,\Sigma_1, \ldots, \Sigma_K})$.
%  Let $a \in \Sigma_i$, $b \in \Sigma_j$, $i < j$, $i,j \in \{1,\ldots,k\}$
%  and $u,v \in \Sigma^*$.
%  Set $(s_1, \ldots, s_k) = \delta(s_0, u)$
%  Then,
%  \begin{align*}  
%   \delta(s_0, uabv) 
%   & = \delta(([\pi_{\Sigma_1}(u)]_{\equiv_L},\ldots, [\pi_{\Sigma_i}(u)]_{\equiv_L},\ldots,[\pi_{\Sigma_j}(u)]_{\equiv_L}, \ldots, [\pi_{\Sigma_k}(u)]_{\equiv_L}, abv) \\
%   & = \delta(([\pi_{\Sigma_1}(u)]_{\equiv_L},\ldots, [\pi_{\Sigma_i}(ua)]_{\equiv_L},\ldots,[\pi_{\Sigma_j}(ub)]_{\equiv_L}, \ldots, [\pi_{\Sigma_k}(u)]_{\equiv_L}, v) \\
%     & = \delta(([\pi_{\Sigma_1}(u)]_{\equiv_L},\ldots, [\pi_{\Sigma_i}(u)]_{\equiv_L},\ldots,[\pi_{\Sigma_j}(u)]_{\equiv_L}, \ldots, [\pi_{\Sigma_k}(u)]_{\equiv_L}, bav) \\
%      & = \delta(s_0, uabv).
%  \end{align*}
%  The property $L(\mathcal C_{L,\Sigma_1,\ldots,\Sigma_k}) \in \mathcal L_{\Sigma_1, \ldots, \Sigma_k}$
%  is clear by construction.\todo{oder die gleichung noch schnell hinschreiben.}
  By Definition~\ref{def:generalization_canonical_aut},
  for each $a \in \Sigma_i$, $b \in \Sigma_j$, $i \ne j$, $i,j \in \{1,\ldots,k\}$
  and state $s \in S_1 \times \ldots \times S_k$,
  we have $\delta(s, ab) = \delta(s, ba)$.
  So, we conclude $L(\mathcal C_{L,\Sigma_1,\ldots,\Sigma_k}) \in \mathcal L_{\Sigma_1, \ldots, \Sigma_k}$. 
  
 \item If $L = L(\mathcal C_{L, \Sigma_1, \ldots,\Sigma_k})$, then, by the previous
 item, $L \in \mathcal L_{\Sigma_1, \ldots, \Sigma_k}$.
 Now, suppose $L \in \mathcal L_{\Sigma_1, \ldots, \Sigma_k}$.
 By the previous item, we only have to establish $L(\mathcal C_{L, \Sigma_1, \ldots,\Sigma_k}) \subseteq L$.
 So, suppose $w \in L(\mathcal C_{L, \Sigma_1, \ldots, \Sigma_k})$.
 By Definition~\ref{def:generalization_canonical_aut},
 there exists $u \in L$ such that $\pi_{\Sigma_i}(w) \equiv_L \pi_{\Sigma_i}(u)$
 for $i \in \{1,\ldots, k\}$.
% Also, by Definition~\ref{def:generalization_canonical_aut}, for any $v \in \Sigma^*$, we have
%  \begin{equation}
%  \label{eqn:C_L}
%   v \in L(\mathcal C_{L, \Sigma_1, \ldots, \Sigma_k})
%   \Leftrightarrow \pi_{\Sigma_1}(v) \pi_{\Sigma_2}(v) \cdots \pi_{\Sigma_k}(v) \in L(\mathcal C_{L, \Sigma_1, \ldots, \Sigma_k})
%  \end{equation}
%  Hence,
%  \[  
%  \pi_{\Sigma_1}(w) \pi_{\Sigma_2}(w) \cdots \pi_{\Sigma_k}(w) \in L(\mathcal C_{L, \Sigma_1, \ldots, \Sigma_k}).
%  \]
 As the letters in $L$ commute according to the partition,
 we have, by Remark~\ref{rem:letters_commutate_with_decomposition}, as $u \in L$,
 \[ 
 \pi_{\Sigma_1}(u) \pi_{\Sigma_2}(u) \pi_{\Sigma_3}(u) \cdots \pi_{\Sigma_k}(u) \in L.
 \]
 So, using $\pi_{\Sigma_1}(w) \equiv_L \pi_{\Sigma_1}(u)$, we find
 \[ 
  \pi_{\Sigma_1}(w) \pi_{\Sigma_2}(u)  \pi_{\Sigma_3}(u)\cdots \pi_{\Sigma_k}(u) \in L.
 \] 
 Using that the letters in $L$ commute according to the partition
 again, we find, with the previous equation, 
 \[ 
  \pi_{\Sigma_2}(u) \pi_{\Sigma_1}(w)  \pi_{\Sigma_3}(u) \cdots \pi_{\Sigma_k}(u) \in L.
 \] 
 And then, using $\pi_{\Sigma_2}(w) \equiv_L \pi_{\Sigma_2}(u)$,
 \[
  \pi_{\Sigma_2}(w) \pi_{\Sigma_1}(w)  \pi_{\Sigma_3}(u) \cdots \pi_{\Sigma_k}(u) \in L.
 \]
 Continuing in this manner, and reordering the result, we get
 \[
  \pi_{\Sigma_1}(w) \pi_{\Sigma_2}(w) \pi_{\Sigma_3}(w) \cdots \pi_{\Sigma_k}(w) \in L.
 \] 
 %So, by Equation~\eqref{eqn:C_L}, $w \in L$.
 So, as the letters in $L$ commute according to the partition, $w \in L$.
 Hence $L(\mathcal C_{L, \Sigma_1, \ldots, \Sigma_k}) \subseteq L$. 
%  Hence $\psi_{\Sigma_1,\ldots, \Sigma_k}(w) = \psi_{\Sigma_1, \ldots, \Sigma_k}(u)$.
%  and we find $L(\mathcal C_{L, \Sigma_1, \ldots, \Sigma_k})
%  \subseteq \psi_{\Sigma_1,\ldots, \Sigma_k}^{-1}(\psi_{\Sigma_1,\ldots, \Sigma_k}(L))$.
%  So, by Remark~\ref{rem:generalized_parikh}, $L(\mathcal C_{L, \Sigma_1, \ldots, \Sigma_k}) \subseteq L$. \qed

 \item As $\mathcal C_{L, \Sigma_1, \ldots \Sigma_k}$ is defined
 with the Nerode right-congruence classes, if $L$
 is regular, it must be a finite automaton. If the automaton is finite
 and $L \in \mathcal L_{\Sigma_1, \ldots, \Sigma_k}$,
 then, by the previous item, $L = L(\mathcal C_{L, \Sigma_1, \ldots, \Sigma_k})$
 and $L$ is regular.
\end{enumerate}
 \noindent So, we have established all statements in the theorem.~\qed
\end{proof}

Also, used in defining a subclass in the next subsection, we will derive a canonical automaton for certain projected languages from $\mathcal C_{L,\Sigma_1,\ldots,\Sigma_k}$. Essentially, the next definition and proposition mean that if we only use  
one ``coordinate'' of $\mathcal C_{L,\Sigma_1, \ldots, \Sigma_k}$, then this recognizes a projection of $L$. %bzw koordinaten zusammenfassen

\begin{definition} \label{def::proj_aut}
Let $i \in \{1,\ldots, k\}$ and $L \in \mathcal L_{\Sigma_1,\ldots,\Sigma_k}$.
The \emph{canonical projection automaton (for $\Sigma_i)$}
is $\mathcal C_{L,\Sigma_i} = (\Sigma_i, S_i, \delta_i, [\varepsilon]_{\equiv_L}, F_i)$
with 
% \begin{align*}
%     S_i & = \{ [u]_{\equiv_L} \mid u \in \Sigma_i^* \}; \\
%     \delta_i([u]_{\equiv_L}, x) & = [ux]_{\equiv_L} \mbox{ for } x \in \Sigma_i; \\
%     F_i & = \{ [\pi_{\Sigma_i}(u)]_{\equiv_L} \mid u \in L \}.
% \end{align*}
$S_i = \{ [u]_{\equiv_L} \mid u \in \Sigma_i^* \}$,
$\delta_i([u]_{\equiv_L}, x)  = [ux]_{\equiv_L} \mbox{ for } x \in \Sigma_i$
and $F_i  = \{ [\pi_{\Sigma_i}(u)]_{\equiv_L} \mid u \in L \}$.
\end{definition}

\begin{propositionrep}
\label{prop:projected_language}
 Let $L \in \mathcal L_{\Sigma_1, \ldots, \Sigma_k}$.
 Then, for $i \in \{1,\ldots,k\}$, $\pi_{\Sigma_i}(L) = L(\mathcal C_{L, \Sigma_i})$.
\end{propositionrep}
\begin{proof} 
 Let $\mathcal C_{L,\Sigma_i} = (\Sigma, S_i, \delta_i, [\varepsilon]_{\equiv_L}, F_i)$.
 Suppose $u \in \pi_{\Sigma_i}(L)$.
 Then, there exists $w \in L$ such that $u = \pi_{\Sigma_i}(w)$.
 So,  $[u]_{\equiv_L} \in F_i$. Hence, $u \in L(\mathcal C_{L, \Sigma_i})$.

 Conversely, if $u \in L(\mathcal C_{L, \Sigma_i}) \subseteq \Sigma_i^*$,
 then $u \equiv_L \pi_{\Sigma_i}(w)$ for some $w \in L$.
 We have\footnote{Set, for $u_1, \ldots, u_n \in \Sigma^*$, $\prod_{i=1}^n u_i = u_1 \cdot \ldots \cdot u_n$.} $\pi_{\Sigma_i}(w)\prod_{j=1, j\ne i}^k \pi_{\Sigma_j}(w) \in L$.
 So, $u \prod_{j=1, j \ne i}^k \pi_{\Sigma_j}(w) \in L$,
 which gives $u \in \pi_{\Sigma_j}(L)$. \qed
\end{proof}

\subsection{Subclasses in  \texorpdfstring{$\mathcal L_{\Sigma_1, \ldots, \Sigma_k}$}{L\_Sigma\_1...Sigma\_k}}
%the Class of Languages whose Letters Commute According to a Given Partition}

Here, we investigate several subclasses of $\mathcal L_{\Sigma_1,\ldots, \Sigma_k}$.
Recall that, for $L \subseteq \Sigma^*$, the minimal automaton
of $L$ is denoted by $\mathcal A_L$.

\begin{definition}
\label{def:Li}
 Let $\Sigma = \Sigma_1 \cup \ldots \cup \Sigma_k$
 be a partition. Then,
 define the following classes of languages.
  \begin{align*}
     % \mathcal L_{\Sigma_1, \ldots, \Sigma_k} & = \{ L \subseteq \Sigma^* \mid \mbox{the letters in $L$ commute
    %  according to the decomposition.} \} \\
      \mathcal L_1 & = \{ L\mid \mathcal C_{L,\Sigma_1,\ldots,\Sigma_k} \mbox{ has a single final state and $L = L(\mathcal C_{L,\Sigma_1,\ldots,\Sigma_k})$. }\}, \\
       \mathcal L_2 & = \left\{ L \mid L = \bigshuffle_{i=1}^k \pi_{\Sigma_i}(L) \right\}, \\
        %\mathcal L_3' & = \left\{ L  \mid \forall i \in \{1,\ldots,k\}\forall u \in \Sigma_i^* : [u]_{\equiv_L} \cap \Sigma_i^* = [u]_{\equiv_{\pi_{\Sigma_i}(L)}} \cap \Sigma_i^* \right\}, \\
        %\mathcal L_3 & = \mathcal L_3' \cap \mathcal L_{\Sigma_1, \ldots, \Sigma_k}, \\
        \mathcal L_3 & = \{ L \mid L = L(\mathcal C_{L,\Sigma_1,\ldots,\Sigma_k}), \forall i \in \{1,\ldots,k\} : \mathcal A_{\pi_{\Sigma_i}(L)} \mbox{ is isomorphic to } \mathcal C_{L,\Sigma_i} \}, \\
      \mathcal L_4 & = \{ L \mid \mathcal A_L \mbox{ is isomorphic to } \mathcal C_{L,\Sigma_1,\ldots,\Sigma_k} \}.
  \end{align*}
\end{definition}

First, we show that these are in fact subclasses of $\mathcal L_{\Sigma_1, \ldots, \Sigma_k}$.

\begin{propositionrep}
\label{prop:L_i_in_L_Sigmai}
 Let $\Sigma = \Sigma_1 \cup \ldots \cup \Sigma_k$
 be a partition.
  For each $i \in \{1,2,3,4\}$
  we have $\mathcal L_i \subseteq \mathcal L_{\Sigma_1, \ldots, \Sigma_k}$.
\end{propositionrep}
\begin{proof}
 If $i \in \{1,3,4\}$, by Theorem~\ref{thm:canonical_aut},
 we have $\mathcal L_i \subseteq \mathcal L_{\Sigma_1,\ldots,\Sigma_k}$.
 By the definition of the shuffle product, $\mathcal L_2  \subseteq \mathcal L_{\Sigma_1, \ldots, \Sigma_k}$.~\qed
%  Let $L \in \mathcal L_3$, $a \in \Sigma_i, b \in \Sigma_j$ with $i \ne j$, $i, j \in \{1,\ldots, k\}$,
%  and $u,v \in \Sigma^*$.
%  Then
%  \begin{align*}
%      [uab]_{\equiv_L} \cap \Sigma_i^* = [\pi_{\Sigma_i}(ua)]_{\equiv_{\pi_{\Sigma_i}(L)}} \cap \Sigma_i^*
%                                       = [uba]_{\equiv_L} \cap \Sigma_i^*.
%  \end{align*}
%  So, $\pi_{\Sigma_i}(ua) \in [uab]_{\equiv_L} \cap [uba]_{\equiv_L}$,
%  which implies $[uab]_{\equiv_L} = [uba]_{\equiv_L}$.
%  As $\equiv_L$ is a right-congruence, $uab \equiv_L uba$
%  implies $uabv \equiv_L ubav$.
%  Combining, this implies that for any equivalence class $[w]_{\equiv_L}$,
%  if $uabv \in [w]_{\equiv_L}$, then $ubav \in [w]_{\equiv_L}$.
%  As $L$ is a union of equivalence classes, this implies $L \in \mathcal L_{\Sigma_1, \ldots, \Sigma_k}$.~\qed
\end{proof}

\begin{remark}
\label{rem:A_L_singlefinal_but_not_C_L}
 Regarding $\mathcal L_1$, note that there exist languages $L = L(\mathcal C_{L,\Sigma_1, \ldots, \Sigma_k})$
 such that the minimal automaton has a single final state, but
 $\mathcal C_{L,\Sigma_1, \ldots, \Sigma_k}$
 has more than one final state.
 For example, $L = \{ w \in \{a,b\}^* \mid |w|_a > 0 \mbox{ or } |w|_b > 0 \}$.
 However, if $\mathcal C_{L,\Sigma_1, \ldots, \Sigma_k}$ has a single final state, then
 the minimal automata also has only a single final state.
\end{remark}

\begin{example}
\label{ex:Lis}
 Let $\Sigma = \Sigma_1 \cup \Sigma_2$
 with $\Sigma_1 = \{a\}$ and $\Sigma_2 = \{b\}$.
 Set $L = ( aa(aaa)^* \shuffle bb(bbb)^* ) \cup ( a(aaa)^* \shuffle b(bbb)^* )$.
 Then $L \in (\mathcal L_3 \cap \mathcal L_4) \setminus \mathcal L_2$.
\end{example}

\begin{example}
 Set $L = ( a(aaa)^* \shuffle b ) \cup aa(aaa)^*$.
 Then $L \in \mathcal L_3 \setminus \mathcal L_4$.
\end{example}

The languages in $\mathcal L_1$ arise in connection
with the canonical automaton.

\begin{propositionrep}
  Let $L \in \mathcal L_{\Sigma_1, \ldots, \Sigma_k}$
  and 
  $
  \mathcal C_{L, \Sigma_1, \ldots, \Sigma_k} = (\Sigma, S_1 \times \ldots \times S_k, \delta, s_0, F).
  $
  Then, for all $s \in S_1 \times \ldots \times S_k$,
  $
   \{ w \in \Sigma^* \mid \delta(s_0, w) = s \} \in \mathcal L_1.
  $
\end{propositionrep}
\begin{proof} 
 Let $s = (s_1, \ldots, s_k) \in S_1 \times \ldots \times S_k$. Set $U = \{ w \in \Sigma^* \mid \delta(s_0, w) = s \}$.
 By construction of $\mathcal C_{L,\Sigma_1, \ldots, \Sigma_k}$, we 
 have $U \in \mathcal L_{\Sigma_1, \ldots, \Sigma_k}$.
 So, by Theorem~\ref{thm:canonical_aut},
 $U = L(\mathcal C_{U, \Sigma_1, \ldots, \Sigma_k})$, where
 $\mathcal C_{U, \Sigma_1, \ldots, \Sigma_k}$ is the canonical automaton for $U$.
 Let 
 \[
  E = \{ ([\pi_{\Sigma_1}(u)]_{\equiv_U}, \ldots, [\pi_{\Sigma_k}(u)]_{\equiv_U}) \mid u \in U \}
 \]
 be the final states of $\mathcal C_{U, \Sigma_1, \ldots, \Sigma_k}$.
 We have to show $|E| = 1$. Let $u_i \in \Sigma_i^*$ be such that $s_i = [u_i]_{\equiv_L}$
 for $i \in \{1,\ldots, k\}$.
 
 \begin{claiminproof}
    For all $v \in \Sigma_i^*$ and $u \in U$ we have
    \[
     v \equiv_U \pi_{\Sigma_i}(u) \Leftrightarrow 
     v \equiv_L u_i.
    \]
 \end{claiminproof}
 \begin{claimproof} We show two separate statements that, taken together, imply
  our claim.
  
 \begin{enumerate} 
  
 \item For all $i \in \{1,\ldots, k\}$ and $u \in U$, $[\pi_{\Sigma_i}(u)]_{\equiv_U} \cap \Sigma_i^* \subseteq [u_i]_{\equiv_L}$.
     
     \medskip

     Fix $i \in \{1,\ldots,k\}$.
     %Observe that for $u \in U$, we
     %have, for any $j \in \{1,\ldots,k\}$, by the transition function of $\mathcal C_{L,\Sigma_1,\ldots,\Sigma_k}$, Definition~\ref{def:generalization_canonical_aut}, $\pi_{\Sigma_j}(u) \equiv_L u_j$.
     Suppose $x \in \Sigma_i^*$
     such that $x \equiv_U \pi_{\Sigma_i}(u)$.
     Then, as 
     \[ 
     \pi_{\Sigma_i}(u) \pi_{\Sigma_1}(u) \cdots \pi_{\Sigma_{i-1}}(u) \pi_{\Sigma_{i+1}}(u)\cdots \pi_{\Sigma_k}(u) \in U,
     \]
     we find, by the definition of the Nerode right-congruence,
     \[
     x \pi_{\Sigma_1}(u) \cdots \pi_{\Sigma_{i-1}}(u) \pi_{\Sigma_{i+1}}(u)\cdots \pi_{\Sigma_k}(u) \in U.
     \]
     Hence, $\delta(s_0, x \pi_{\Sigma_1}(u) \cdots \pi_{\Sigma_{i-1}}(u) \pi_{\Sigma_{i+1}}(u)\cdots \pi_{\Sigma_k}(u)) = s$.
     %So, as $x \in \Sigma_i^*$ and by the transition function of $\mathcal C_{L,\Sigma_1,\ldots,\Sigma_k}$, $x \equiv_L u_i$.
     By Remark~\ref{rem:can_aut_transition}, as $x \in \Sigma_i^*$,
     we find $x \equiv_L u_i$.
     
     \medskip 
     
 \item For all $i \in \{1,\ldots, k\}$, if $x,y \in [u_i]_{\equiv_L}\cap \Sigma_i^*$,
  then $x \equiv_U y$.
  
    \medskip 
    
    Fix $i \in \{1,\ldots,k\}$.
    Assume there exist $x,y \in \Sigma_i^*$ such that $x \not\equiv_U y$, but $x,y \in [u_i]_{\equiv_L}$.
    Then, without loss of generality, there exists $z \in \Sigma^*$
    such that 
    \[
     xz \in U, \quad yz \notin U.
    \]
    But then, for all $j \in \{1,\ldots, k\}$, $\pi_{\Sigma_j}(xz) \equiv_L u_j$.
    As $x,y \in \Sigma_i^*$, for $j \ne i$, we also have $\pi_{\Sigma_j}(yz) \equiv_L u_j$.
    So, for $yz \notin U$ to hold true, we must have $\pi_{\Sigma_i}(yz) \not\equiv_L u_i$.
    However, by assumption $x,y \in  [u_i]_{\equiv_L}$, so
    \[
     \pi_{\Sigma_i}(yz) = y \pi_{\Sigma_i}(z) \equiv_L x \pi_{\Sigma_i}(z) = \pi_{\Sigma_i}(xz) \equiv_L u_i,
    \]
    which implies $\pi_{\Sigma_i}(yz) \equiv_L u_i$. So, $yz \notin U$ is not possible.
    Hence, all words from $\Sigma_i^*$ in $[u_i]_{\equiv_L}$
    must be equivalent for $\equiv_U$. 
 \end{enumerate}
 Combining the first and second part gives, for each $i \in \{1,\ldots, k\}$ and $u \in U$, $[\pi_{\Sigma_i}(u)]_{\equiv_U} \cap \Sigma_i^* = [u_i]_{\equiv_L} \cap \Sigma_i^*$.
 \end{claimproof}
 Now, choose $u, u' \in U$ and $i \in \{1,\ldots, k\}$.
 Then, by the above claim, we find
 \[
  \pi_{\Sigma_i}(u) \equiv_L u_i \mbox{ and }
  \pi_{\Sigma_i}(u') \equiv_L u_i.
 \]
 Using the above claim, with $v = \pi_{\Sigma_i}(u')$,
 we can then deduce $\pi_{\Sigma_i}(u') \equiv_U \pi_{\Sigma_i}(u)$.
 So, for each $u, u' \in U$, we have
 \[
  ([\pi_{\Sigma_1}(u)]_{\equiv_U}, \ldots, [\pi_{\Sigma_k}(u)]_{\equiv_U})
  = ([\pi_{\Sigma_1}(u')]_{\equiv_U}, \ldots, [\pi_{\Sigma_k}(u')]_{\equiv_U}),
 \]
 which implies $|E| = 1$.~\qed
\end{proof}

Next, we give alternative characterization for $\mathcal L_2, \mathcal L_3$
and $\mathcal L_4$.

\begin{theoremrep}
\label{thm:alternative_characterizations}
 Let $L \in \mathcal L_{\Sigma_1, \ldots, \Sigma_k}$. Then,
 \begin{enumerate}
 \item $L \in \mathcal L_2$ if and only if, for each $w \in \Sigma^*$,
 the following is true:
 \[
  w \in L \Leftrightarrow \forall i \in \{1,\ldots,k\} : \pi_{\Sigma_i}(w) \in \pi_{\Sigma_i}(L);
 \]
 
 \item  $
   L \in \mathcal L_3$ if and only if, for all $i \in \{1,\ldots,k\}$ and $u \in \Sigma_i^*$,
   we have 
  \[ 
  [u]_{\equiv_L} \cap \Sigma_i^* = [u]_{\equiv_{\pi_{\Sigma_i}(L)}} \cap \Sigma_i^*;
  \]
 
 \item $L \in \mathcal L_4$ if and only if, for each $u,v \in \Sigma^*$,
 \[ 
  u \equiv_L v 
  \Leftrightarrow \forall i \in \{1,\ldots,k\} : \pi_{\Sigma_i}(u) \equiv_L \pi_{\Sigma_i}(v).
 \]
 \end{enumerate}
\end{theoremrep}
\begin{proof}
 The separate claims are stated in
 Proposition~\ref{prop:char_L2}, Proposition~\ref{prop:char_L3}
 and Proposition~\ref{prop:char_L4}.\qed
\end{proof}

\begin{toappendix}

Note that, for each $L \subseteq \Sigma^*$,
\begin{equation}
\label{eqn:L_in_bigshuffle_pi_Sigma_i}
    L \subseteq \bigshuffle_{i=1}^k \pi_{\Sigma_i}(L).
\end{equation}
%and, by Lemma~\ref{lem:generalized_parik},
%$\bigshuffle_{i=1}^k \pi_{\Sigma_i}(L) \in \mathcal L_{\Sigma_1, \ldots, \Sigma_k}$.

\begin{proposition}
\label{prop:char_L2} 
% Let $\Sigma = \Sigma_1 \cup \ldots \cup \Sigma_k$ be a decomposition.
 Let $L \in \mathcal L_{\Sigma_1, \ldots, \Sigma_k}$.
 Then, $L \in \mathcal L_2$ if and only if, for each $w \in \Sigma^*$,
 the following is true:
 \[
  w \in L \Leftrightarrow \forall i \in \{1,\ldots,k\} : \pi_{\Sigma_i}(w) \in \pi_{\Sigma_i}(L).
 \]
\end{proposition}
\begin{proof} 
\begin{enumerate}
\item  Suppose $L = \bigshuffle_{i=1}^k \pi_{\Sigma_i}(L)$.
 We show that the equivalence holds true.
 If $w \in L$, then, for each $i \in \{1,\ldots,k\}$, $\pi_{\Sigma_i}(w) \in \pi_{\Sigma_i}(L)$.
 Conversely, let $w \in \Sigma^*$ and assume, for each $i \in \{1,\ldots,k\}$, we have
 $\pi_{\Sigma_i}(w) \in \pi_i(L)$.
 Then, $w \in \bigshuffle_{i=1}^k \pi_{\Sigma_i}(w) \subseteq \bigshuffle_{i=1}^k \pi_{\Sigma_i}(L) = L$.

\item Now, suppose, for each $w \in \Sigma^*$, we have
 \[
  w \in L \Leftrightarrow \forall i \in \{1,\ldots,k\} : \pi_{\Sigma_i}(w) \in \pi_i(L).
 \]
 By Equation~\eqref{eqn:L_in_bigshuffle_pi_Sigma_i}, $L \subseteq \bigshuffle_{i=1}^k \pi_{\Sigma_i}(L)$.
 So, assume $w \in \bigshuffle_{i=1}^k \pi_{\Sigma_i}(L)$.
 Fix $j \in \{1,\ldots,k\}$.
 Then, as $\Sigma_1,\ldots,\Sigma_k$ are pairwise disjoint
 and so $\pi_{\Sigma_i}(\pi_{\Sigma_j}(L)) = \{\varepsilon\}$ for $i \ne j$,
 \[
 \pi_{\Sigma_j}(w) \in \pi_{\Sigma_j}\left( \bigshuffle_{i=1}^k \pi_{\Sigma_i}(L) \right) = 
 \bigshuffle_{i=1}^k \pi_{\Sigma_j}(\pi_{\Sigma_i}(L)) = \pi_{\Sigma_j}(L).
 \]
 Hence, by applying the equation, we find $w \in L$.
 So, $\bigshuffle_{i=1}^k \pi_{\Sigma_i}(L) \subseteq L$.
\end{enumerate}
 So, we have shown the equivalence of both conditions.~\qed
\end{proof}

\begin{proposition}
\label{prop:char_L3}
  Let $L \in \mathcal L_{\Sigma_1,\ldots,\Sigma_k}$.
  Then,
  $
   L \in \mathcal L_3$ if and only if, for each $i \in \{1,\ldots,k\}$ and $u \in \Sigma_i^*$,
   we have 
    $[u]_{\equiv_L} \cap \Sigma_i^* = [u]_{\equiv_{\pi_{\Sigma_i}(L)}} \cap \Sigma_i^*.
  $
\end{proposition}
\begin{proof}
 The map $\varphi : \{ [u]_{\equiv_L} \cap \Sigma_i^* \mid u \in \Sigma_i^* \} \to 
 \{ [u]_{\pi_{\Sigma_i}(L)} \cap \Sigma_i^* \mid u \in \Sigma_i^* \}$ given by 
 \[
  \varphi([u]_{\equiv_L} \cap \Sigma_i^*) = [u]_{\equiv_{\pi_{\Sigma_i}(L)}} \cap \Sigma_i^*
 \]
 is well-defined by Lemma~\ref{lem:equivalence_classes_projection}.
 Also, the non-empty sets $[u]_{\equiv_L} \cap \Sigma_i^*$
 partition $\Sigma_i^*$ as a right-congruence over $\Sigma_i^*$.
 Moreover, the non-empty right-congruence classes of the form $[u]_{\equiv_L} \cap \Sigma_i^*$
 are precisely the states of $\mathcal C_{L, \Sigma_i}$.
 So, $\varphi$ induces a surjective homomorphism
 from $\mathcal C_{L,\Sigma_i}$ onto $\mathcal A_{\pi_{\Sigma_i}(L)}$.
 This yields the stated equivalence.~\qed
\end{proof}

\begin{proposition} 
\label{prop:char_L4}
 Let %$\Sigma = \Sigma_1 \cup \ldots \cup \Sigma_k$ and 
 $L \in \mathcal L_{\Sigma_1,\ldots,\Sigma_k}$.
 Then, $L \in \mathcal L_4$ if and only if, for each $u,v \in \Sigma^*$,
 \[ 
  u \equiv_L v 
  \Leftrightarrow \forall i \in \{1,\ldots,k\} : \pi_{\Sigma_i}(u) \equiv_L \pi_{\Sigma_i}(v).
 \]
\end{proposition}
\begin{proof} 
 First, we show that the implication from right to left is always true. % oder das als gesonderte aussage?
 Hence, we only have to argue that if the other implication
 is true, this is equivalent to $L \in \mathcal L_4$.
 
 \begin{claiminproof}
  Let $L \in \mathcal L_{\Sigma_1, \ldots, \Sigma_k}$ and $u,v \in \Sigma^*$. Then,
  \[
    \forall i \in \{1,\ldots,k\} : \pi_{\Sigma_i}(u) \equiv_L \pi_{\Sigma_i}(v) \Rightarrow u \equiv_L v.
  \]
 \end{claiminproof}
 \begin{claimproof}
 Suppose \[ \forall i \in \{1,\ldots,k\} : \pi_{\Sigma_i}(u) \equiv_L \pi_{\Sigma_i}(v). \]
 Then, using $\pi_{\Sigma_1}(u) \equiv_L \pi_{\Sigma_1}(v)$, at it is a right-congruence, we find
 \[
  \pi_{\Sigma_1}(u) \pi_{\Sigma_2}(v) \cdots \pi_{\Sigma_k}(v)
   \equiv_L
  \pi_{\Sigma_1}(v) \pi_{\Sigma_2}(v) \cdots \pi_{\Sigma_k}(v).
 \]
 As $L \in \mathcal L_{\Sigma_1, \ldots, \Sigma_k}$, we have
 \[ 
 \pi_{\Sigma_1}(u) \pi_{\Sigma_2}(v)\pi_{\Sigma_3}(v)  \cdots \pi_{\Sigma_k}(v) \equiv_L \pi_{\Sigma_2}(v)\pi_{\Sigma_1}(u)\pi_{\Sigma_3}(v) \cdots \pi_{\Sigma_k}(v).
 \]
 Then, using $\pi_{\Sigma_2}(u) \equiv_L \pi_{\Sigma_2}(v)$,
 \[
  \pi_{\Sigma_2}(u) \pi_{\Sigma_1}(u)\pi_{\Sigma_3}(v) \cdots \pi_{\Sigma_k}(v)
   \equiv_L
  \pi_{\Sigma_2}(v) \pi_{\Sigma_1}(u)\pi_{\Sigma_3}(v) \cdots \pi_{\Sigma_k}(v).
 \]
 Hence, up to now,
 \[
   \pi_{\Sigma_2}(u) \pi_{\Sigma_1}(u)\pi_{\Sigma_3}(v) \cdots \pi_{\Sigma_k}(v)
   \equiv_L
    \pi_{\Sigma_2}(v) \pi_{\Sigma_1}(v)\pi_{\Sigma_3}(v) \cdots \pi_{\Sigma_k}(v).
 \]
 Continuing similarly, we can show that
 \[
   \pi_{\Sigma_2}(u) \cdots \pi_{\Sigma_k}(u)
   \equiv_L
   \pi_{\Sigma_2}(u) \cdots \pi_{\Sigma_k}(u).
 \]
 So\footnote{Alternatively, we can note that by commutativity, the Nerode right-congruence is also
 a left-congruence. Hence, it is the syntactic congruence, giving a natural composition operation
 on the equivalence classes. Then, 
 by Remark~\ref{rem:letters_commutate_with_decomposition},
 \[ 
 [u]_{\equiv_L} = [\pi_{\Sigma_1}(u) \cdots \pi_{\Sigma_k}(u)]_{\equiv_L}
  = [\pi_{\Sigma_1}(u)]_{\equiv_L} \cdots [\pi_{\Sigma_k}(u)]_{\equiv_L},
 \]
 which also gives that, if for all $i \in \{1,\ldots,k\}$
 we have $\pi_{\Sigma_i}(u) \equiv_L \pi_{\Sigma_i}(v)$, then $u \equiv_L v$.}, 
 by Remark~\ref{rem:letters_commutate_with_decomposition},
 $u \equiv_L v$. 
 \end{claimproof}
 
 Therefore, as
 \[
  \forall i \in \{1,\ldots,k\} : \pi_{\Sigma_i}(u) \equiv_L \pi_{\Sigma_i}(v) \Rightarrow 
  u \equiv_L v,
 \]
 the map
 \[
  ([\pi_{\Sigma_1}(u_1)]_{\equiv_L}, \ldots, [\pi_{\Sigma_k}(u_k)]_{\equiv_L}) \mapsto [\pi_{\Sigma_1}(u_1) \cdots \pi_{\Sigma_k}(u_k)]_{\equiv_L}.
 \]
 for $u_1, \ldots, u_k \in \Sigma^*$
 is well-defined. For, suppose $v_1, \ldots, v_k \in \Sigma^*$
 such that
 \[
  \forall i \in \{1,\ldots,k\} : \pi_{\Sigma_i}(u_i) \equiv_L \pi_{\Sigma_i}(v_i).
 \]
 Set $u = \pi_{\Sigma_1}(u_1)\cdots\pi_{\Sigma_k}(u_k)$
 and $v = \pi_{\Sigma_1}(v_1)\cdots\pi_{\Sigma_k}(v_k)$.
 Then, for all $i \in \{1,\ldots,k\}$,
 we have $\pi_{\Sigma_i}(u) = \pi_{\Sigma_i}(u_i)$
 and $\pi_{\Sigma_i}(v) = \pi_{\Sigma_i}(v_i)$,
 which implies $\pi_{\Sigma_i}(u) \equiv_L \pi_{\Sigma_i}(v)$
 and we find that $u \equiv_L v$.
 Also, as $([\pi_{\Sigma_1}(u)]_{\equiv_L}, \ldots, [\pi_{\Sigma_k}(u)])$
 gets mapped to $[\pi_{\Sigma_1}(u)\cdots\pi_{\Sigma_k}(u)]_{\equiv_L} = [u]_{\equiv_L}$,
 the map is surjective.
 By commutativity, we can also easily show that it is an automaton
 homomorphism.

 So, if $L \in \mathcal L_4$, then $\mathcal C_{L,\Sigma_1,\ldots, \Sigma_k}$
 and $\mathcal A_L$ have the same number of states. Hence, the above given
 surjective map is actually bijective.
 So,
 \[
  [\pi_{\Sigma_1}(u_1) \cdots \pi_{\Sigma_k}(u_k)]_{\equiv_L}
  = [\pi_{\Sigma_1}(v_1) \cdots \pi_{\Sigma_k}(v_k)]_{\equiv_L}
 \]
 implies 
 \[
  ([\pi_{\Sigma_1}(u_1)]_{\equiv_L}, \ldots, [\pi_{\Sigma_k}(u_k)]_{\equiv_L}) 
  = ([\pi_{\Sigma_1}(v_1)]_{\equiv_L}, \ldots, [\pi_{\Sigma_k}(v_k)]_{\equiv_L}). 
 \] 
 In particular, if $u \equiv_L v$,
 then, by Remark~\ref{rem:letters_commutate_with_decomposition},
 we have $[\pi_{\Sigma_1}(u) \cdots \pi_{\Sigma_k}(u)]_{\equiv_L}
  = [\pi_{\Sigma_1}(v) \cdots \pi_{\Sigma_k}(v)]_{\equiv_L}$,
  which gives
 \[ 
 ([\pi_{\Sigma_1}(u)]_{\equiv_L}, \ldots, [\pi_{\Sigma_k}(u)]_{\equiv_L}) 
  = ([\pi_{\Sigma_1}(v)]_{\equiv_L}, \ldots, [\pi_{\Sigma_k}(v)]_{\equiv_L}) 
 \] 
 and we have the implication
 \begin{equation}\label{eqn:L_4_implication}
  u \equiv_L v \Rightarrow 
  \forall i \in \{1,\ldots,k\} : \pi_{\Sigma_i}(u) \equiv_L \pi_{\Sigma_i}(v).
 \end{equation}
 
 Conversely, suppose Equation~\eqref{eqn:L_4_implication}
 holds true.
 Then, the map
 \[ 
  [u]_{\equiv_L}  \mapsto ([\pi_{\Sigma_1}(u)]_{\equiv_L} , \ldots, [\pi_{\Sigma_k}(u)]_{\equiv_L} )
 \] 
 is well-defined.
 As, for $u_1, \ldots, u_k \in \Sigma^*$,
 $u = \pi_{\Sigma_1}(u_1) \cdots \pi_{\Sigma_k}(u_k)$
 gets mapped to $([\pi_{\Sigma_1}(u_1)]_{\equiv_L} , \ldots, [\pi_{\Sigma_k}(u_k)]_{\equiv_L} )$,
 it is also surjective. So, by finiteness\footnote{Actually, we do not need to use
 finitess, and can show it directly by noting that the given map is a two-sided inverse
 to the map
 \[
  ([\pi_{\Sigma_1}(u_1)]_{\equiv_L}, \ldots, [\pi_{\Sigma_k}(u_k)]_{\equiv_L}) \mapsto [\pi_{\Sigma_1}(u_1) \cdots \pi_{\Sigma_k}(u_k)]_{\equiv_L}.
 \]}, it is a bijection between $\mathcal A_L$
 and $\mathcal C_{L,\Sigma_1, \ldots, \Sigma_k}$.~\qed
\end{proof}
\end{toappendix}

\begin{toappendix}
\begin{remark}
 Note that the condition in Proposition~\ref{prop:char_L3}
 is stated only for words from $\Sigma_i^*$, $i \in \{1,\ldots,k\}$.
 For example, stipulating the equation
 \[ [u]_{\equiv_L} \cap \Sigma_i^* = [\pi_{\Sigma_i}(u)]_{\equiv_{\pi_{\Sigma_i}(L)}} \cap \Sigma_i^*
 \]
 for all $u \in \Sigma^*$ is a strong demand.
 It implies, for $a \in \Sigma_i$, $b \in \Sigma_j$, $i \ne j$,
 $[ab]_{\equiv_L} = [a]_{\equiv_L} = [b]_{\equiv_L} $,
 as
 \begin{align*} 
  [ab]_{\equiv_L} \cap \Sigma_i^* & = [a]_{\equiv_{\pi_{\pi_{\Sigma_i}(L)}}} \cap \Sigma_i^*; \\ 
  [ab]_{\equiv_L} \cap \Sigma_j^* & = [b]_{\equiv_{\pi_{\pi_{\Sigma_j}(L)}}} \cap \Sigma_j^*, 
 \end{align*}
 which gives $a,b \in [ab]_{\equiv_L}$.
\end{remark}

\begin{remark}
 Let $L \subseteq \Sigma^*$.
 In several statements, we have intersected the equivalence classes for $\pi_{\Sigma_i}(L)$, $i \in \{1,\ldots,k\}$,
 with $\Sigma_i^*$.
 Equivalently, we could also say that we only consider the equivalence over the smaller alphabet $\Sigma_i^*$.
 More generally, if $\Gamma \subseteq \Sigma$ and $U \subseteq \Gamma^*$,
 then, $u, v \in \Gamma^*$ are equivalent for $\equiv_U$ over the alphabet $\Gamma$
 if and only if they are equivalent for $\equiv_U$ over $\Sigma$.
 Also, over the larger alphabet, all words using symbols not in $\Gamma$
 are equivalent.
\end{remark}

\begin{remark}
% \todo[inline]{Klasse von Sprachen, für die Projektionen klein. siehe auch structural property paper?}
 The condition in Proposition~\ref{prop:char_L3}
 does not imply that the language is in $\mathcal L_{\Sigma_1, \ldots, \Sigma_k}$.
 Let $L = \{ ab, b\}$. 
 Then, $\pi_{\{a\}}(L) = \{a\}$, $\pi_{\{b\}}(L) = \{b\}$ and, for each $n,m \ge 0$,
 $
     [a^n]_{\equiv_L} \cap a^*  = [a^n]_{\equiv_{\{a\}}}$ and 
     $[b^m]_{\equiv_L} \cap b^*  = [b^m]_{\equiv_{\{b\}}}.$
 However, $L \notin \mathcal L_{\{a\},\{b\}}$.
\end{remark}
\end{toappendix}

\begin{example}
\label{ex:L4-incomparable}
 Let $L_1$ be the language from Example~\ref{ex::comm_aut}.
 Set $L_2 = a_1 \shuffle a_2 = \{a_1 a_2, a_2 a_1\}$.
 Both of their letters commute for the partition $\{a_1,a_2\} = \{a_1\} \cup \{a_2\}$.
 Then, $L_1 \in \mathcal L_4 \setminus \mathcal  L_3$ and $L_2 \in \mathcal L_1 \setminus \mathcal L_4$.
%  \begin{enumerate}
%      \item $L_1 \in \mathcal L_4 \setminus \mathcal  L_3$.
%      \item $L_2 \in \mathcal L_1 \setminus \mathcal L_4$.
%  \end{enumerate}
\end{example}

Finally, in Theorem~\ref{thm:L1inL2inL3},
we establish inclusion relations, which are all proper, between
$\mathcal L_1, \mathcal L_2$ and $\mathcal L_3$, also see Figure~\ref{fig:inclusions_Li}.

\begin{figure}[htb]
     \centering
\begin{tikzpicture}[scale=.7]
 \node (top) at (0,2)  {$\mathcal L_{\Sigma_1, \ldots, \Sigma_k}$};
 \node (L3)  at (-2,1) {$\mathcal L_3$};
 \node (L2)  at (-2.5,0) {$\mathcal L_2$};
 \node (L1)  at (-2,-1) {$\mathcal L_1$};
 \node (bottom) at (0,-2) {$\{\emptyset, \Sigma^*\}$};
 \node (L4)  at (1,0) {$\mathcal L_4$};
 
 \draw (top) -- (L3) -- (L2) -- (L1) -- (bottom) -- (L4) -- (top);
\end{tikzpicture}
  \caption{Inclusion relations between the language classes. %, see Proposition~\ref{prop:L_i_in_L_Sigmai}, Theorem~\ref{thm:L1inL2inL3} and Remark~\ref{rem:inclusions}.
  }
  \label{fig:inclusions_Li}
\end{figure}

\begin{toappendix}

\begin{example}\label{ex::lang_classes} 
 For the language $L = (a(aaa)^* \cup aa(aaa)^*) \shuffle b = a(aaa)^* \shuffle b \cup aa(aaa)^* \shuffle b$
 the minimal commutative automaton has more than a single
 final state, but $L = \pi_1(L) \shuffle \pi_2(L)$.
\end{example}

\begin{proposition} \label{prop::general_spec_form_final_states}
 Let $L \in \mathcal L_{\Sigma_1,\ldots,\Sigma_k}$.
 Then, the following statements are equivalent:
 \begin{enumerate}
 \item $L = \bigshuffle_{i=1}^k \pi_{\Sigma_i}(L)$;
 
 \item in the canonical automaton $\mathcal C_{L,\Sigma_1, \ldots, \Sigma_k} = (\Sigma, S_1 \times \ldots \times S_k, \delta, s_0, F)$
  we can write $F = F_1 \times \ldots \times F_k$ with $F_i \subseteq S_i$.
 \end{enumerate}
\end{proposition}
\begin{proof} 
 \begin{enumerate}
 \item Suppose $L = \bigshuffle_{i=1}^k \pi_{\Sigma_i}(L)$.
  For the final state set $F$ of the canonical automaton, we will
  show $F = F_1 \times \ldots \times F_k$.
     
  \medskip
  
  First, note that, by Definition~\ref{def:generalization_canonical_aut},
  we find $F \subseteq F_1 \times \ldots \times F_k$.
  For the converse inclusion, let $u_1, \ldots, u_k \in L$, and consider the tuple 
  \[ ( [\pi_{\Sigma_1}(u_1)]_{\equiv_L}, \ldots, [\pi_{\Sigma_1}(u_k)]_{\equiv_L} ) \in F_1 \times \ldots \times F_k. 
  \]
  Choose $u \in \prod_{i=1}^k \pi_{\Sigma_i}(u_i)$.
  By assumption, $\prod_{i=1}^k \pi_{\Sigma_i}(L) \subseteq L$.
  Hence $u \in L$.
  As $\pi_{\Sigma_i}(u) = \pi_{\Sigma_i}(u_i)$ for $i \in \{1,\ldots, k\}$,
  we find 
  \[
   ( [\pi_{\Sigma_1}(u_1)]_{\equiv_L}, \ldots, [\pi_{\Sigma_1}(u_k)]_{\equiv_L} ) 
  = ( [\pi_{\Sigma_1}(u)]_{\equiv_L}, \ldots, [\pi_{\Sigma_1}(u)]_{\equiv_L} ). 
  \]
  By Definition~\ref{def:generalization_canonical_aut},
  $( [\pi_{\Sigma_1}(u)]_{\equiv_L}, \ldots, [\pi_{\Sigma_1}(u)]_{\equiv_L} ) \in F$.
  So, $F_1 \times \ldots \times F_k \subseteq F$.

 \item Assume $F = F_1 \times \ldots \times F_k$.
 
  \medskip 
 
 Let $u \in \bigshuffle_{i=1}^k \pi_{\Sigma_i}(L)$.
 Then there exist words $u_i \in L$, $i \in \{1,\ldots, k\}$,
 such that  $\pi_{\Sigma_i}(u) = \pi_{\Sigma_i}(u_i)$.
 So, $[\pi_{\Sigma_i}(u_i)]_{\equiv_L} \in F_i$
 and we find 
 \[ 
  ([\pi_{\Sigma_1}(u)]_{\equiv_L}, \ldots, [\pi_{\Sigma_k}(u)]_{\equiv_L}) \in F.
 \]
 By Definition~\ref{def:generalization_canonical_aut}, there exists $v \in L$
 such that
 \[ 
 ([\pi_{\Sigma_1}(u)]_{\equiv_L}, \ldots, [\pi_{\Sigma_k}(u)]_{\equiv_L})
 = ([\pi_{\Sigma_1}(v)]_{\equiv_L}, \ldots, [\pi_{\Sigma_k}(v)]_{\equiv_L}).
 \]
 Using Remark~\ref{rem:letters_commutate_with_decomposition} and the previous equation,
 \begin{align*}
     [u]_{\equiv_L} & = [\pi_{\Sigma_1}(u) \cdots \pi_{\Sigma_k}(u)]_{\equiv_L} \\
                    & = [\pi_{\Sigma_1}(u)]_{\equiv_L} \cdot\ldots\cdot [\pi_{\Sigma_k}(u)]_{\equiv L} \\
                    & = [\pi_{\Sigma_1}(v)]_{\equiv_L} \cdot\ldots\cdot [\pi_{\Sigma_k}(v)]_{\equiv L} \\
                    & = [v]_{\equiv L}.
 \end{align*}
 So, as $v \in L$, we have $u \in L$.
 Hence, $\prod_{i=1}^k \pi_{\Sigma_i}(L) \subseteq L$.
 %Conversely, as, for any $w \in \Sigma^*$,
 %we have $w \in \prod_{i=1}^k \pi_{\Sigma_i}(w)$,
 %we find $L \subseteq \prod_{i=1}^k \pi_{\Sigma_i}(L)$.
 So, with Equation~\eqref{eqn:L_in_bigshuffle_pi_Sigma_i}, 
 $L = \prod_{i=1}^k \pi_{\Sigma_i}(L)$.
 \end{enumerate}
 Hence, the first and the last condition are equivalent
 and this concludes the proof.~\qed
\end{proof}
\end{toappendix}

\begin{theoremrep}
\label{thm:L1inL2inL3}
  We have $\mathcal L_1 \subsetneq \mathcal L_2 \subsetneq \mathcal L_3$.
\end{theoremrep}
\begin{proof}
  By Proposition~\ref{prop::general_spec_form_final_states}, 
  we find $\mathcal L_1 \subseteq \mathcal L_2$. That the inclusion is proper is shown
  by Example~\ref{ex::lang_classes}. 
  Suppose $L = \bigshuffle_{i=1}^k \pi_{\Sigma_i}(L)$.
  Let $i \in \{1,\ldots,k\}$ and $u,v \in \Sigma_i^*$.
  We have 
  \[
   u \equiv_L v  \Leftrightarrow \forall x \in \Sigma^* : ux \in L \leftrightarrow vx \in L.
  \] 
  Using Proposition~\ref{prop:char_L2}, the condition on the right hand side is equivalent to
  \begin{multline}\label{eqn:L_2}
         \forall x \in \Sigma^* [ ( \forall j \in \{1,\ldots, k\} : \pi_{\Sigma_j}(ux) \in \pi_{\Sigma_j}(L) ) \leftrightarrow \\ ( \forall j \in \{1,\ldots, k\} : \pi_{\Sigma_j}(vx) \in \pi_{\Sigma_j}(L) ) ].
  \end{multline}
  
  \begin{claiminproof}
    Equation~\eqref{eqn:L_2} is equivalent to
    \begin{equation} \label{eqn:L_2_2}
     \forall x \in \Sigma^* : \pi_{\Sigma_i}(ux) \in \pi_{\Sigma_i}(L) \leftrightarrow \pi_{\Sigma_i}(vx) \in \pi_{\Sigma_i}(L).
    \end{equation}
  \end{claiminproof}
  \begin{claimproof}
   Suppose Equation~\eqref{eqn:L_2_2} holds true.
   Let $x \in \Sigma^*$.
   Then, if, for each $j \in \{1,\ldots,k\}$,
   we have 
   \[
     \pi_{\Sigma_j}(ux) \in \pi_{\Sigma_j}(L),
   \]
   using that, for $u,v \in \Sigma_i^*$, we have,
  \[
   \pi_{\Sigma_j}(ux) = \left\{ 
   \begin{array}{ll}
    u \pi_{\Sigma_i}(x) & \mbox{if } i = j; \\
      \pi_{\Sigma_j}(x) & \mbox{if } i \ne j,
   \end{array}\right.
  \]
  and similarly for $\pi_{\Sigma_j}(vx)$,
  we find with Equation~\eqref{eqn:L_2_2}
  that, for each $j \in \{1,\ldots,k\}$,
   \[
     \pi_{\Sigma_j}(ux) \in \pi_{\Sigma_j}(L) ).
   \]
   The other implication of Equation~\eqref{eqn:L_2}
   can be shown similarly. Hence,  Equation~\eqref{eqn:L_2} holds true.
   
   Conversely, suppose now Equation~\eqref{eqn:L_2}
   holds true.
   Let $x \in \Sigma^*$
   and assume
   \[
    \pi_{\Sigma_i}(ux) \in \pi_{\Sigma_i}(L).
   \]
   As $L \ne \emptyset$,
   we find $u_j \in \Sigma_j^*$
   such that $u_j \in \pi_{\Sigma_j}(L)$
   for each $j \in \{1,\ldots,k\}$.
   Then, set $y = u_1 \cdots u_{i-1} \pi_{\Sigma_i}(x) u_{i+1}\cdots u_k$.
   By choice and as $\pi_{\Sigma_i}(y) =  \pi_{\Sigma_i}(x)$,
   we have for each $j \in \{1,\ldots,k\}$ that
   \[
    \pi_{\Sigma_j}(u y) \in \pi_j(L).
   \] 
   Hence, by Equation~\eqref{eqn:L_2}, for each $j \in \{1,\ldots,k\}$,
   \[
    \pi_{\Sigma_j}(v y) \in \pi_j(L).
   \]
   In particular, $\pi_{\Sigma_i}(vx) = \pi_{\Sigma_i}(vy) \in \pi_i(L)$.
   The other implication of Equation~\eqref{eqn:L_2_2}
   can be shown similarly. Hence,  Equation~\eqref{eqn:L_2_2} holds true.
  \end{claimproof}

  So, by the previous claim, Equation~\eqref{eqn:L_2}
  simplifies to
  \[
   \forall x \in \Sigma^* : \pi_{\Sigma_i}(ux) \in \pi_{\Sigma_i}(L) \leftrightarrow \pi_{\Sigma_i}(vx) \in \pi_{\Sigma_i}(L),
  \] 
  which is equivalent to
  \[ 
  \forall x \in \Sigma_i^* : ux \in \pi_{\Sigma_i}(L) \leftrightarrow vx \in \pi_{\Sigma_i}(L).
  \]
  But this is precisely the definition of $u \equiv_{\pi_{\Sigma_i}(L)} v$.
  Hence,
  \[
    [u]_{\equiv_L} \cap \Sigma_i^* = [u]_{\pi_{\Sigma_i}(L)} \cap \Sigma_i^*
  \]
  and by Proposition~\ref{prop:char_L3}, we find $\mathcal L_2 \subseteq \mathcal L_3$.
  That the inclusion is proper, is shown by Example~\ref{ex:Lis}.~\qed
\end{proof}

\begin{remark}
\label{rem:inclusions}
 Theorem~\ref{thm:L1inL2inL3} and Example~\ref{ex:L4-incomparable}
 show that $\mathcal L_4$ is incomparable to each of the other language classes
 with respect to inclusion.
 %,
 %i.e., for $i \in \{1,2,3\}$, we have neither $\mathcal L_4 \subseteq \mathcal L_i$
 %nor $\mathcal L_i \subseteq \mathcal L_4$.
\end{remark}

\section{Conclusion} The language class of commutative regular languages
with minimal automata of product-form behaves well with respect
to the descriptional complexity measure of state complexity for certain operations, see Table~\ref{tab:sc_product-from},
and Lemma~\ref{lem:sc_shuffle_lang} allows us to construct
infinitely many commutative regular languages with product-form minimal automaton.
%, while
%still being rich in the sense that many commutative and regular languages fall into this class.
The investigation started could be carried out for other operations and measures of descriptional complexity
as well. Likewise, as done in~\cite{GomezA08,DBLP:conf/icgi/Gomez10} 
for commutative and more general partial commutativity conditions,
it might be interesting if the learning algorithms given there could be improved
for the language class introduced.

Lastly, if the bound $2nm$ for shuffle is tight is  an open problem.
Remark~\ref{rem:lower_bound_shuffle} shows that the bound $nm$ is not sufficient,
however, giving an infinite family of commutative regular languages with minimal automata
of product-form attaining the bound $2nm$ for shuffle is an open problem.

{\smallskip \noindent \footnotesize
\textbf{Acknowledgement.} I thank the anonymous referees of~\cite{Hoffmann2021NISextended} (the extended version of~\cite{DBLP:conf/cai/Hoffmann19}), whose feedback also helped in the present work. I also sincerely
thank the referees of the present submission, which helped me alot
in identifying unclear or ungrammatical formulations and
a missing definition.

%indirectly resulted in the investigations
%carried out in the present work.
}
\bibliographystyle{splncs04}
\bibliography{ms} 
\end{document}